\documentclass[12pt]{amsart}
\usepackage[a4paper,margin=2.9cm]{geometry}                
\usepackage{graphicx,verbatim}
\usepackage{amssymb,cite}
\usepackage{epstopdf}
\usepackage{graphicx}
\usepackage{amsmath}
\usepackage{mathrsfs}
\usepackage{caption}
\usepackage{subcaption}
\usepackage{color}
\usepackage[normalem]{ulem}   

\usepackage{float}
\floatstyle{boxed}
\newfloat{alg}{ht}{alg}

\title{Electronic Density of States for Incommensurate Layers}
\author{Daniel Massatt, Mitchell Luskin, Christoph Ortner}
\date{\today}                                           
\thanks{DM was supported by NSF PIRE Grant OISE-0967140. ML was supported in
  part by ARO MURI Award W911NF-14-1-0247 and by the Radcliffe Institute for
  Advanced Study at Harvard University. CO was supported by ERC Starting Grant
  335120.}

\numberwithin{equation}{section}

\newtheorem{remark}{Remark}

\newtheorem{thm}{Theorem}
\newtheorem{lemma}{Lemma}

\newtheorem{assumption}{Assumption}
\numberwithin{definition}{section}
\numberwithin{thm}{section}
\numberwithin{remark}{section}
\numberwithin{prop}{section}
\numberwithin{corollary}{section}
\numberwithin{lemma}{section}
\numberwithin{assumption}{section}

\newcommand{\Tr}{\text{Tr}}

\newcommand{\C}{\mathcal{C}}

\newcommand{\R}{\mathcal{R}}

\newcommand{\Z}{\mathbb{Z}^2}

\newcommand{\A}{\mathcal{A}}

\newcommand{\D}{\mathcal{D}}
\newcommand{\MatSpace}{ M_{|\Omega_r|}(\mathbb{C})}
\newcommand{\M}{H}

\newcommand{\Ainv}{\mathfrak{A}}
\newcommand{\Rinv}{\mathfrak{R}}

\newcommand{\Mat}{\tilde{H}}
\newcommand{\Msup}{S[H]}
\newcommand{\DiscSize}{N_{\text{disc}}}
\def\mod{{\rm mod}}
\newcommand{\DiscSpace}{S}

\newcommand{\Imaginary}{\text{Im}}
\newcommand{\Resolvent}{\text{Res}}
\newcommand{\DoS}{\text{DoS}}

\begin{document}

\begin{abstract}
   We prove that the electronic density of states (DOS) for 2D incommensurate layered
   structures, where Bloch theory does not apply, is well-defined as the
   thermodynamic limit of finite clusters. In addition, we obtain an
   explicit representation formula for the DOS as an integral over
   local configurations.

   Next, based on this representation formula, we propose a novel algorithm for
   computing electronic structure properties in incommensurate heterostructures,
   which overcomes limitations of the common approach to artificially strain a
   large supercell and then apply Bloch theory.
\end{abstract}

\maketitle

\section{Introduction}

Bloch theory provides an elegant solution for describing the electronic
structure of periodic materials. However, there has been a lot of focus recently
on the study of {\em incommensurate} layers of two-dimensional crystal
structures \cite{Terrones2014,2DPerturb15}. In the absence of periodicity,
computing the electronic structure of such materials becomes more challenging.

A common approach to approximate the electronic properties of such a system is
to artificially strain it to obtain periodicity on a large supercell, and then
apply Bloch theory to this periodic
system~\cite{Terrones2014,Loh2015,Ebnonnasir2014,Koda2016,Komsa2013}.
Commensurate approximations to an incommensurate system are computationally
expensive, and their approximation error is unclear. Here we introduce a new
method for computing a class of observables derived from the density of states
for multi-layer incommensurate heterostructures {\em without} requiring an
artificial strain in the system.

To approximate an observable of an infinite incommensurate system, we
approximate local lattice site contributions to the observable. We observe that
a site is uniquely defined by its local geometry. Using an equidistribution
theorem, there is a predictable distribution of local geometries, and hence site
contributions. Consequently, we can express observables in incommensurate
heterostructures in terms of an integral over a unit cell, in a fashion rather
similar to Bloch theory. This unit cell classification of local configurations
is related to Bellisard's noncommutative Brillouin Zone for aperiodic
solids~\cite{Bellissard2002}. Prodan used the Bellisard formalism to compute
electronic properties for periodic materials with on-site defects modeled by a
tight-binding model~\cite{prodan2012}. Here we consider the density of states
and related observables for incommensurate multi-layers.

While the methodology is in principle generic, our derivation and analysis
focuses on tight-binding models, which are commonly employed for computing the
electronic structure of 2D materials~\cite{Castro2009,Kaxiras_2003}. We
consider the density of states and related observables for incommensurate
multi-layers. We use Chebyshev Polynomial methods to approximate the density
of states as a function~\cite{Huang:2006go, Mazzi:2011hl, Roder:1997ee,
Silver:1996uj, kernel_poly}, and from this function any observable can be
computed.

\subsection*{Outline}
In Section \ref{sec:main} we introduce the results for the
bilayer case, and briefly discuss their extension to the multi-layer case. In
Section \ref{subsec:inc} we introduce incommensurate systems and
the equidistribution result. In Section \ref{subsec:tb} we specify
the details of our model problem, and in Section \ref{subsec:local} we show how
to compute the local density of states. In Section \ref{subsec:therm} we prove
the infinite system is well posed and express the observables as an integral
over local observables.

Section \ref{sec:numerics} we describe an approximation scheme and present
numerical results. In Section \ref{subsec:disc} we discuss the integral
discretization. In Section \ref{subsec:approx}, we introduce a Chebyshev Kernel Polynomial Method, and in Section \ref{subsec:num_results} we present numerical
results. In Section \ref{sec:thm} we present the details of the proofs.


\section{Main Results}
\label{sec:main}

\subsection{Incommensurate Heterostructures}
\label{subsec:inc}
Consider two periodic atomic sheets in parallel 2D planes separated by a constant distance. Each individual sheet can be described as a Bravais lattice embedded in $\mathbb{R}^2$ by neglecting the out of plane distance. This coordinate is not relevant for classifying the aperiodicity and will be incorporated in section \S~ \ref{subsec:tb}. For sheet $j \in \{1,2\}$, we define the Bravais lattice
\begin{equation*}
\R_j = \{ A_jn \text{ : } n \in \mathbb{Z}^2\},
\end{equation*}
where $A_j$ is a $2\times2$ invertible matrix. We define the {\em unit cell} for sheet $j$ as
\begin{equation*}
\Gamma_j = \{ A_j \alpha \text{ : } \alpha \in [0,1)^2\}.
\end{equation*}
Each individual sheet is trivially periodic, since
\begin{equation*}
\R_j = A_j n+ \R_j \qquad \text{for $n \in \mathbb{Z}^2$. }
\end{equation*}
However, the combined system $\R_1\cup\R_2$ need not be periodic (Figure
\ref{fig:incommMoire}). (Note that here $\R_1 \cup \R_2$ is only considered to
describe geometry, not as an indexing of the atoms as
it would have the failure of identifying the origins from each lattice.)

Since we are interested in aperiodic systems, we make the following standing assumption:

\begin{assumption}
\label{assump:main}
The lattices $\R_1$ and $\R_2$ are {\em incommensurate}, that is, for $v \in \mathbb{R}^2$,
\begin{equation*}
v + \R_1 \cup \R_2 = \R_1 \cup \R_2 \hspace{2mm}\Leftrightarrow \hspace{2mm}v = \begin{pmatrix} 0\\0\end{pmatrix}.
\end{equation*}
\end{assumption}
\begin{figure}[ht]
\centering
\begin{subfigure}{.45\textwidth}
\centering
\includegraphics{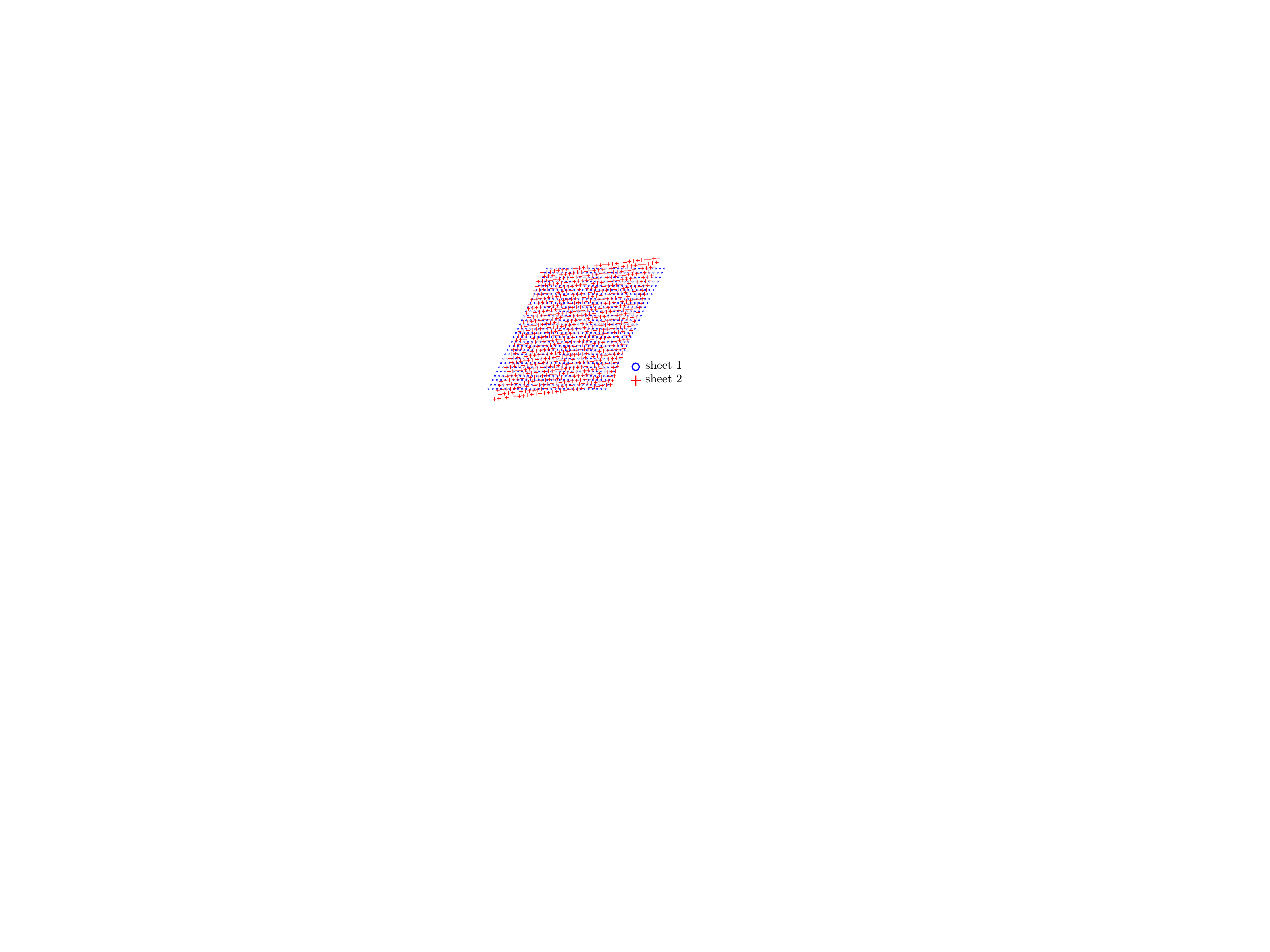}
\caption{An incommensurate hexagonal bilayer.  Sheet 1 is rotated by $\theta = 6^\circ$ relative to sheet 2.}
\label{fig:incommMoire}
\end{subfigure}
\hspace{2mm}
\begin{subfigure}{.45\textwidth}
\centering
\vspace{4mm}
\includegraphics{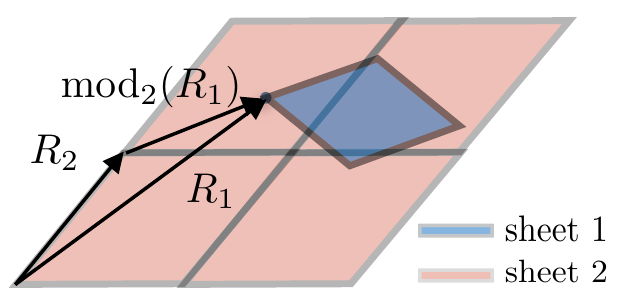}
\vspace{4mm}
\caption{$\mod_2(R_1)$ is the shift of the first lattice relative to the second lattice.}
\label{fig:bshift}
\end{subfigure}
\caption{Visualisation of incommensurate bilayer geometry.}
\end{figure}

Since the majority of material simulation tools rely on periodicity, the most
common method at present to simulate incommensurate layers is to adjust one of
the two layers slightly in order to make the system commensurate on some larger
supercell (Figure \ref{fig:incomm_comm}). In contrast we take advantage of an
equidistribution of local geometries.

To parameterize the local geometries, we define the modulation operator $\mod_j : \mathbb{R}^2 \rightarrow \Gamma_j$ on sheet $j$ for position $u \in \mathbb{R}^2$:
\begin{equation*}
\mod_j(u) := u + R_j \, \text{ where } R_j \in \R_j \text{ such that } u + R_j \in \Gamma_j.
\end{equation*}
Then the relative shift of site $R_1 \in \R_1$ is $\mod_2(R_1) \in \Gamma_2$ (See Figure \ref{fig:bshift}). The local geometry of site $R_1 \in \R_1$ is defined by
\begin{equation*}
\R_1\cup\R_2 - R_1 = \R_1 \cup ( \R_2-R_1) = \R_1 \cup(\R_2-\mod_2(R_1)).
\end{equation*}
 Hence, the local geometry is determined by the relative shift $\mod_2(R_1)$. The
same argument holds for relative configurations around a site on sheet two. A
fundamental idea in this method is that the distribution of $\mod_j(R_{P_j}) \in
\Gamma_j$ is uniform in the sense of Theorem~\ref{thm:equidistribution} below.

We let
\begin{equation*}
B_r = \{ y \in \mathbb{R}^2 : |y| < r\}, \qquad \text{for $r > 0$.}
\end{equation*}
For $j \in\{1,2\},$ we let $P_j$ be the transposition, that is, $P_1=2$ and $P_2=1.$

\begin{figure}[ht]
\centering
\begin{subfigure}{.4\textwidth}
\centering
\includegraphics[width=.8\linewidth]{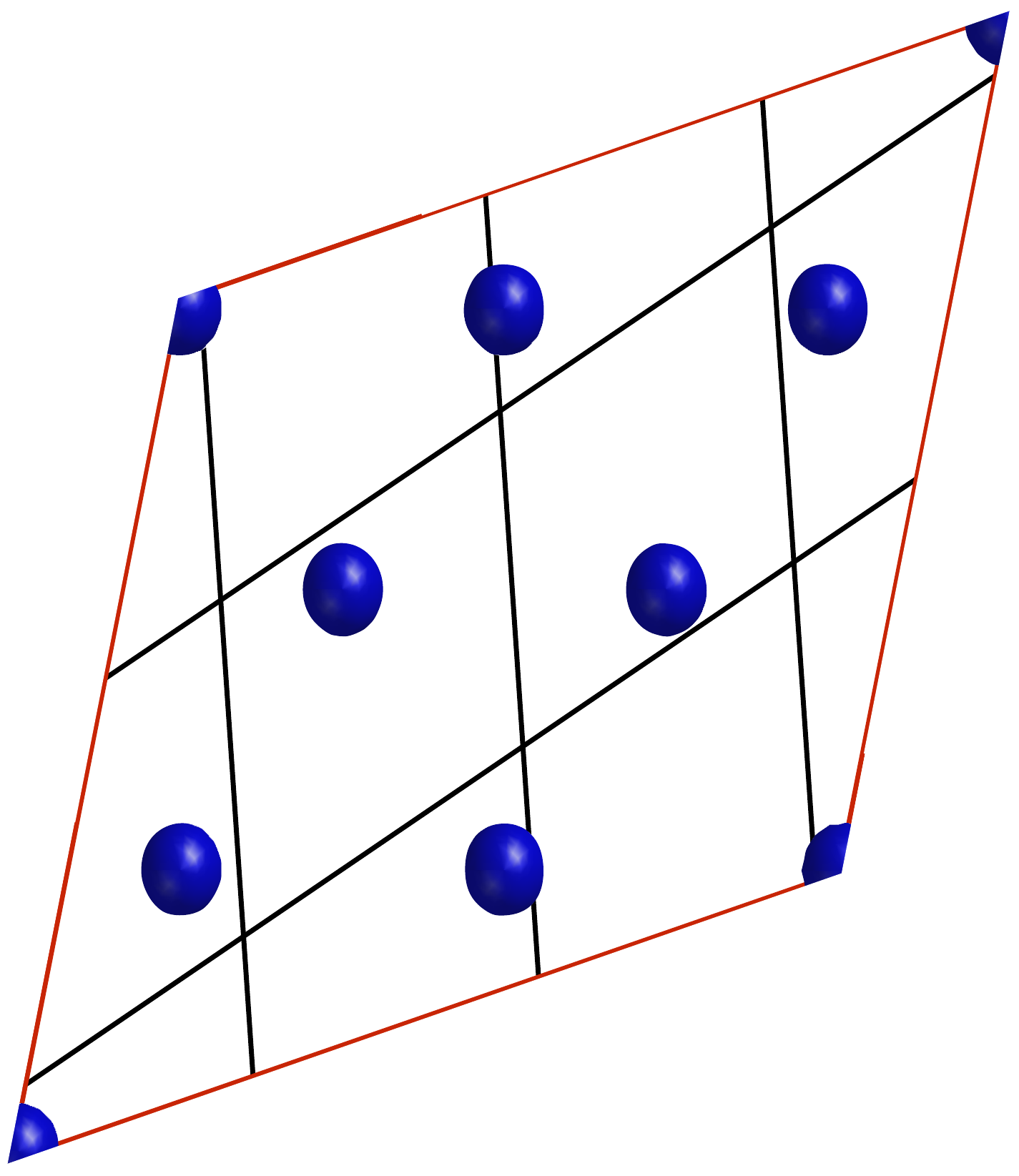}
\caption {Incommensurate Cell}
\label{fig:incomm}
\end{subfigure}
\begin{subfigure}{.4\textwidth}
\centering
\includegraphics[width=.765\linewidth]{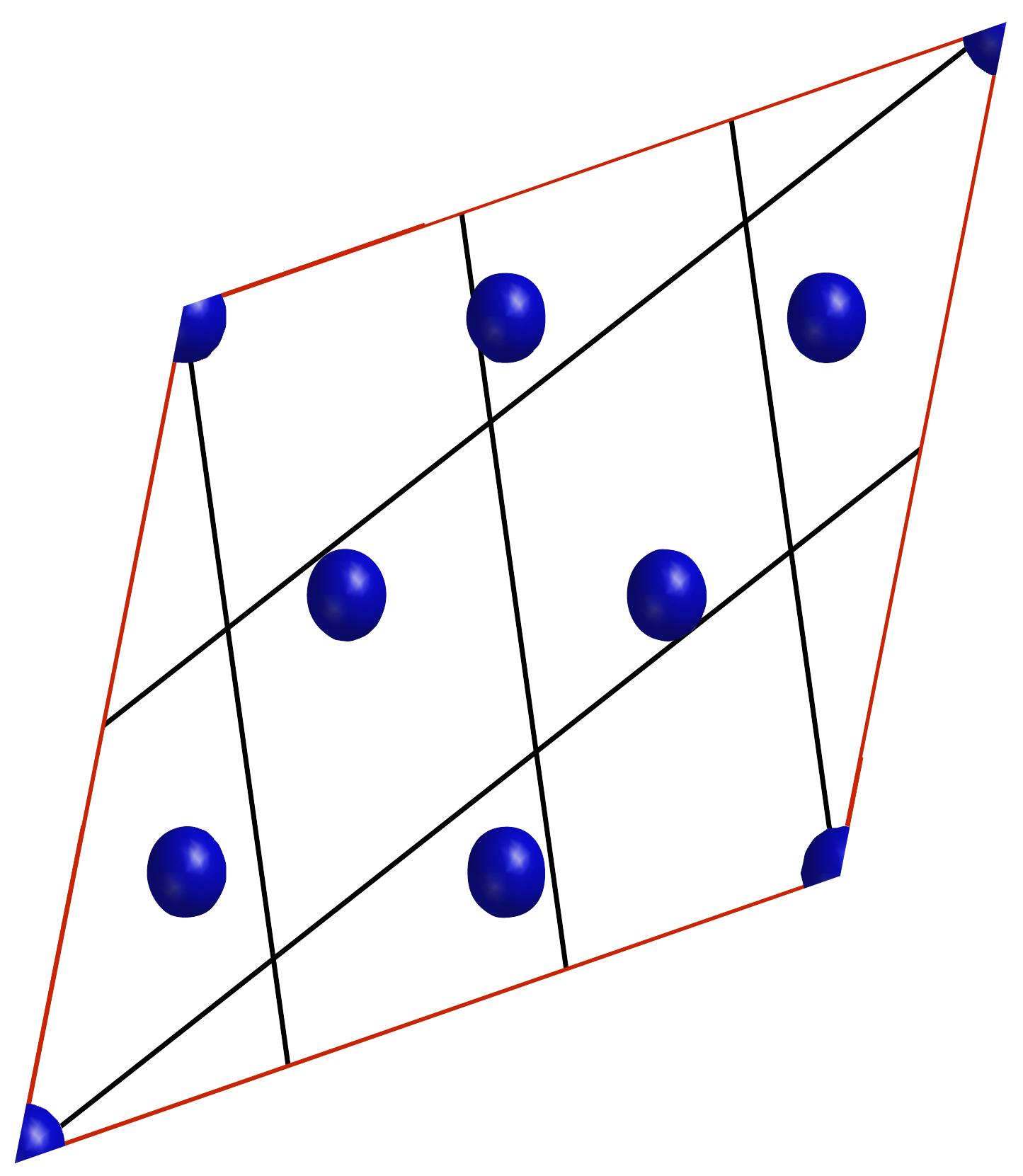}
\caption{ Commensurate Cell Approximation }
\label{fig:comm}
\end{subfigure}
   \caption{(A) Two lattices (spheres and lines) that are incommensurate;
   (B) The sphere lattice is slightly rotated to obtain a commensurate
   cell approximation.}
   \label{fig:incomm_comm}
\end{figure}

\begin{thm}
\label{thm:equidistribution}
Consider $\R_1$ and $\R_2$ incommensurate lattices embedded in $\mathbb{R}^2$
(i.e., satisfying Assumption \ref{assump:main}). Then for $g \in C_{\text{per}}(\Gamma_{P_j})$, we have
\begin{equation} \label{eq:equidistribution}
\frac{1}{\# \R_j \cap B_r} \sum_{\ell \in \R_j \cap B_r} g(\ell) \rightarrow \frac{1}{|\Gamma_{P_j}|}\int_{\Gamma_{P_j}} g(b)db.
\end{equation}
\end{thm}

In particular, local geometries around sheet 1 sites can be parameterized by
$\Gamma_2$, while local geometries around sheet 2 sites can be parameterized by
$\Gamma_1$.

Theorem \ref{thm:equidistribution} suggests the following strategy for
defining {\em and computing} electronic structure properties in incommensurate
heterostructures: (1) Split an observable into local contributions
from each atomic site (we will employ the local density of states);
(2) Employ Theorem~\ref{thm:equidistribution} to
demonstrate that the thermodynamic limit from finite clusters exist
(observe that \eqref{eq:equidistribution} is a sum over a finite cluster);
(3) Use the right-hand side of \eqref{eq:equidistribution} to compute
the limit quantity.

\subsection{Tight-Binding Model}
\label{subsec:tb}
Electronic structure is governed by solutions to the Schr{\" o}dinger
eigenproblem. It is typically approximated using methods such as the Kohn--Sham DFT (KS-DFT) model or the
Hartree--Fock approximation \cite{Kaxiras_2003}. For systems in the thousands of
atoms however, the standard KS-DFT calculation becomes intractable. The
tight-binding (TB) model applies further approximations, and as a result can treat larger systems
ranging in the millions of atoms.

Let $\A_i$ denote the set of indices of orbitals associated with each unit cell
of sheet $i$. We assume that $\A_i$ are finite and that $\A_1 \cap \A_2 =
\emptyset$. Then the full degree of freedom space is
\begin{equation*}
   \Omega = (\R_1 \times \A_1)\cup (\R_2 \times \A_2).
\end{equation*}
The interaction between orbitals indexed by $R\alpha$ and $R'\alpha'$ is denoted by $h_{\alpha\alpha'}(R-R')$, where $h_{\alpha\alpha'} \in C(\mathbb{R}^2)$. Although the sheets have a vertical displacement between them,
this distance is constant and hence can be encoded into $h_{\alpha\alpha'}$
(using the assumption that $\A_1 \cap \A_2 = \emptyset$).
We will further use the following assumption:
\begin{assumption}
 \label{assump:decay}
  Orbital interactions $h_{\alpha\alpha'}$ are uniformly continuous on $\mathbb{R}^2$ and decay exponentially, that is,
   \begin{equation*}
   |h_{\alpha\alpha'}(x)| \leq Ce^{-\tilde{\gamma}|x|} \qquad \text{for } x \in \mathbb{R}^2.
   \end{equation*}
\end{assumption}
This applies in most scenarios, since in most tight-binding models the orbitals are {\em tightly bound} around the atomic sites \cite{Kaxiras_2003}, or  are exponentially decaying.
We then formally define a matrix $H$ such that
\begin{equation*}
H_{R\alpha,R'\alpha'} = h_{\alpha\alpha'}(R-R').
\end{equation*}
This is an infinite matrix, hence the eigenproblem
\begin{equation*}
H\psi = E\psi
\end{equation*}
for $\psi \in \mathbb{C}^\mathbb{N}$ cannot be solved directly.
Instead, we will define a class of observables for the infinite system by first
defining them for finite sub-systems and then passing to the limit in \S~
\ref{subsec:therm}.

For $\tilde\Omega \subset \Omega$ with $\# \tilde\Omega = n$ the associated
hamiltonian is $\Mat = (H_{ij})_{i,j \in \tilde\Omega} \in M_n(\mathbb{C})$,
where $M_n(\mathbb{C})$ denotes the set of $n \times n$ Hermitian matrices over $\mathbb{C}$.
The {\em density of states} for $\tilde\Omega$ can be defined via its action
on test functions, or, observables $g$, by
\begin{equation*}
   \D[\Mat](g) = \frac{1}{n} \Tr [g(\Mat)], \hspace{2mm} g \in C(\mathbb{R}).
\end{equation*}
(We will later slightly extend the space of observables.)
For example, we can consider the bond energy $\D[\Mat](U_T)$, where $U_T(\epsilon) = \epsilon F_T(\epsilon)$ and
   $F_T(\epsilon) = (1+e^{(\epsilon-\mu)/kT})^{-1}$
 is the Fermi function. Formally, the value of the observable for
the infinite system $\Omega$ is the limit of $\D[\Mat](g)$ as $\tilde\Omega \uparrow \Omega$.

For future reference we remark that, since $H$ is defined in terms of the
lattices $R_j$ and the hopping functions $h_{\alpha\alpha'}$, we will say
``$H$ satisfies Assumptions~\ref{assump:decay} and~\ref{assump:main}'' to mean
that ``$R_1, R_2$ satisfy Assumption~\ref{assump:main} and $h_{\alpha\alpha'}$
satisfy Assumption~\ref{assump:decay}''.

\subsection{Local Density of States}
\label{subsec:local}
%
The next step is to define the local density of states distribution, which will allow us to identify local site contribution to an observable.
Consider a finite sub-system $\tilde{\Omega} \subset \Omega$
with associated hamiltonian $\Mat \in M_n(\mathbb{C})$, then the local density of states distribution is defined as
\begin{equation*}
\D_k[\Mat](g) = [g(\Mat)]_{kk}, \qquad k \in \tilde\Omega, \hspace{1mm}g \in C(\mathbb{R}).
\end{equation*}
Note that
\begin{equation*}
   \frac{1}{n}\sum_{k\in\tilde\Omega} \D_k[\Mat](g) = \D[\Mat](g).
\end{equation*}
This reformulation puts us very close to the setting of Theorem~\ref{thm:equidistribution}.
It remains to control the dependence of $\D_k[\Mat](g)$ on $\tilde{\Omega}$, which
we will achieve in the next section by fixing $k$ and letting $\tilde\Omega \uparrow \Omega$
while controlling the error.

\begin{figure}[ht]
\centering
\includegraphics[width=.5\linewidth]{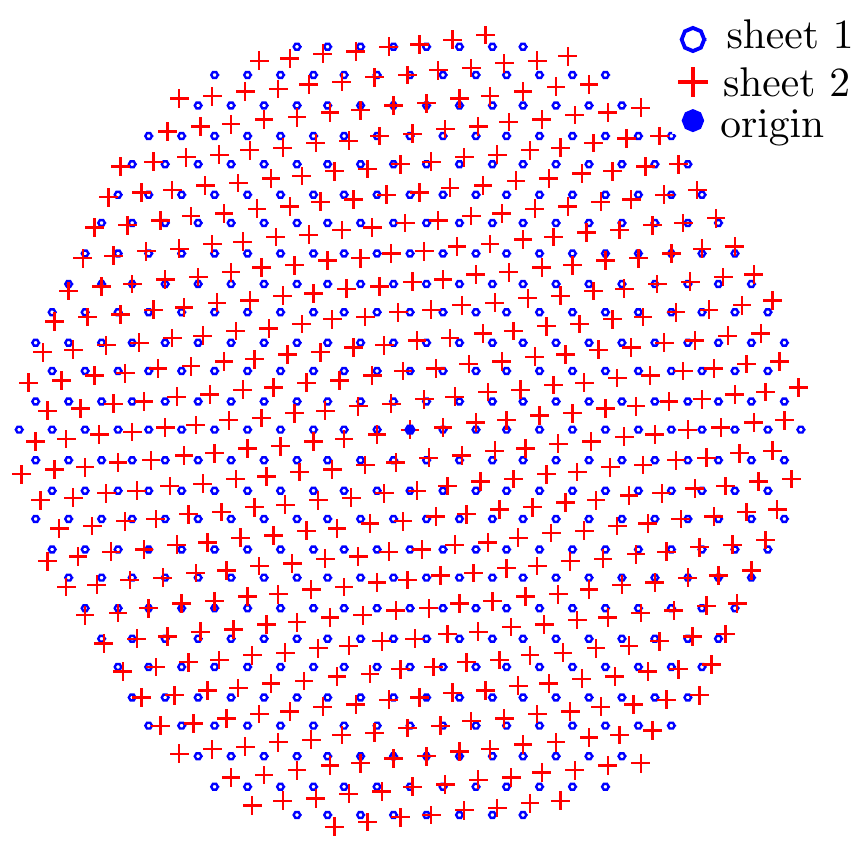}
\caption{All the sites in $\Omega_r$ for a hexagonal bravais lattice. The central site for sheet 1 is highlighted. }
\label{fig:locality}
\end{figure}

Towards that end we now specify a sequence of local degree of freedom spaces,
\begin{equation*}
   \Omega_r = \bigl[\R_1 \cap B_r\bigl]\times \mathcal{A}_1
            \, \cup \, \bigl[\R_2 \cap B_r\bigl]\times \mathcal{A}_2,
            \qquad \text{ for } r > 0;
\end{equation*}
see also Figure \ref{fig:locality}. For $r > 0$ and $b \in \mathbb{R}^2$ we define
 $H_{r,j}(b) \in \MatSpace$ by
\begin{equation*}
[H_{r,j}(b)]_{R\alpha,R'\alpha'} = h_{\alpha\alpha'}\bigl(b(\delta_{\alpha \in \mathcal{A}_{P_j}} - \delta_{\alpha' \in \mathcal{A}_{P_j}})+R-R'\bigr),
\end{equation*}
for $R\alpha, R'\alpha' \in \Omega_r$.
Physically, $H_{r,j}(b)$ describes a cluster of radius $r$ of the bilayer
system in which the sheet $P_j$ is shifted by $b$.
The local configuration is determined by the relative shift, so $b$ indexes which local configuration we are considering. Then
\begin{equation}
   \label{eq:H_rj_b}
\D_\alpha[H_{r,j}(b)]:=\D_{0\alpha}[H_{r,j}(b)]\quad\text{for}\quad\alpha \in \A_j
 \end{equation}
 is an approximate local density of states distribution of the infinite system at a local configuration indexed by $b \in \Gamma_{P_j}$ at orbital $\alpha$ on sheet $j$.

\subsection{Thermodynamic Limit}
\label{subsec:therm}
We now consider the limit as $r \rightarrow \infty$ of
the LDoS, which will allow us to define the DoS for the infinite system. Let
\begin{equation*}
   E[\M] := \sup_{r > 0,\, j \in \{1,2\}}\biggl[\sup_{b \in \Gamma_j} \|H_{r,j}(b)\|_2 \biggr]< \infty,
\end{equation*}
where $\| \tilde{H} \|_2 := \sup_{\psi \in \mathbb{C}^n \setminus \{0\}} \|\tilde{H}\psi\|_2 / \| \psi \|_2$, for $\tilde{H} \in M_n(\mathbb{C})$.
Then the local density of states distribution will be supported on the interval
\begin{equation*}
   \Msup = [-E[\M],E[\M]].
\end{equation*}
We can now generalize observables to be functions $g \in C(\Msup)$ and supply this space with the norm
\begin{equation*}
\|g\|_\infty := \sup_{x \in \Msup}|g(x)|, \qquad \text{for } g \in C(\Msup).
\end{equation*}

For $U \subset \mathbb{C},$ we define the distance
\begin{equation*}
d(U,\Msup) =\inf_{z \in U, z' \in \Msup} |z-z'|.
\end{equation*}
This is a bound on the distance between $U$ and the spectrum.
To pass to the limit in the LDoS and later in the DoS, we narrow down admissible test functions to
\begin{equation*}
   \Lambda := \{ g \in C(\mathbb{R}) ~|~ \text{ $g$ is analytic on $\Msup$} \}.
\end{equation*}
If $g \in \Lambda$, then there exists $\tilde{d} > 0$ such that $g \in
\Lambda_{\tilde{d}}$, which is defined as
\begin{equation*}
   \Lambda_{\tilde{d}} := \{ g \in C(\mathbb{R}) ~|~ \text{ $g$ is analytic at $z$ for } d(z, \Msup) \leq \tilde{d}\}.
\end{equation*}

\begin{thm}
\label{thm:conv}
 (1) Suppose that $H$ satisfies Assumptions~\ref{assump:main} and~\ref{assump:decay}. Then, for $\alpha \in \mathcal{A}_j$, there exists a function $\D_\alpha[\M] : \Gamma_{P_j} \times C(\Msup) \rightarrow \mathbb{C}$ such that, for $g \in \Lambda$,
\begin{equation*}
   \D_{\alpha}[H_{r,j}(b)](g) \rightarrow \D_\alpha[\M](b,g) \hspace{2mm} \text{ as } r\rightarrow \infty.
\end{equation*}
(The distribution $\D_\alpha[\M](b,g)$ is the local density of states for the infinite system.)

(2) The map $g \mapsto \D_{\alpha}[\M](b,g)$ is a bounded linear functional, more precisely,
\begin{equation*}
\bigl|\D_{\alpha}[\M](b,g)\bigr| \leq \|g\|_\infty \qquad \text{for }
g \in C(\Msup).
\end{equation*}

(3) There exist constants $C, \gamma' > 0$ such that, for $\tilde{d} > 0$
   and $g \in \Lambda_{\tilde{d}}$,
   \begin{equation*}
      |\D_\alpha[\M](b,g) - \D_\alpha[H_{r,j}(b)](g)|
      \leq C\tilde{d}^{-2} {\sup_{d(z,\Msup) < \tilde{d}}} |g(z)| e^{-\gamma' \tilde{d} r}.
   \end{equation*}
\end{thm}

We next analyze the regularity of the map $b \mapsto \D_\alpha[\M](b,g)$ for
fixed $g$, which will allow us to integrate with respect to $b$. Let
$n\in\mathbb{Z},\, m = (m_1, m_2) \in\mathbb{N}^2$ such that $m_1 + m_2 \leq
n$. Then, for $f \in C^n(\mathbb{R}^2)$, we employ the usual multi-index
notation
\begin{equation*}
   \partial_m f = \frac{\partial^{m_1+m_2} f}{\partial x_1^{m_1} \partial x_2^{m_2}}.
\end{equation*}

\begin{thm}
\label{thm:smooth}
Suppose $h_{\alpha\alpha'} \in C^n(\mathbb{R}^2)$ for $n \in \mathbb{N}\cup\{0, \infty\}$,
 $\partial_{b_1}^m\partial_{b_2}^{m'}h_{\alpha\alpha'}$ is uniformly continuous
 for $m+m' \leq n$ and satisfies
\begin{equation*}
|\partial_{b_1}^m\partial_{b_{2}}^{m'} h_{\alpha\alpha'}(r)| \leq Ce^{-\gamma'r}.
\end{equation*}
 Then, for $\alpha \in \mathcal{A}_j$ and $g \in \Lambda$,
\begin{equation*}
\D_\alpha[H](\cdot,g) \in C_{\rm per}^n(\Gamma_{P_j}).
\end{equation*}
\end{thm}

Our next objective is to rigorously define the density of states distribution
for the infinite incommensurate bilayer system $\M$. Taking a sequence of finite
incommensurate clusters surrounded by vacuum that grow towards infinity and
combining our results on the equidistribution of local configurations with the
convergence of the local density of states we obtain the following
representation formula.

\begin{thm}
\label{thm:thermal}
   Suppose that $H$ satisfies Assumptions~\ref{assump:main} and~\ref{assump:decay}. Then there exists a bounded linear functional $\D[\M] : C(\Msup) \rightarrow \mathbb{C}$ such that, for $g \in \Lambda$, we have
\begin{equation*}
\D[H_{r,j}(0)](g)
\rightarrow \D[\M](g) \hspace{2mm} \text{ as } r\rightarrow\infty,\quad\text{for }j=1,\,2,
\end{equation*}
and
\begin{equation*}
\D[\M](g) = \nu \sum_{j=1}^2\sum_{\alpha \in \mathcal{A}_j} \int_{\Gamma_{P_j}} \D_\alpha[\M](b,g)db,
\end{equation*}
where
\begin{equation*}
\nu = \frac{1}{|\mathcal{A}_2|\cdot|\Gamma_1|+|\mathcal{A}_1|\cdot|\Gamma_2|}.
\end{equation*}

If $g \in \Lambda_{\tilde{d}}$, then we have the explicit error bound
\begin{equation*}
\biggl|\D[\M](g) - \nu \sum_{j=1}^2\sum_{\alpha \in \mathcal{A}_j}  \int_{\Gamma_{P_j}}\D_{\alpha}[H_{r,j}(b)](g)db\biggr| \leq C\tilde{d}^{-2}
   \sup_{d(z,\Msup) < \tilde{d}} |g(z)|e^{-\gamma\tilde{d} r},
\end{equation*}
where $C, \gamma$ are independent of $r, \tilde{d}$ and $g$.
\end{thm}

\begin{remark}
The finite systems employed in the thermodynamic limit are defined by the matrices $H_{r,j}(0)$ for $j=1,2.$ They represent finite incommensurate clusters surrounded by vacuum. Since the boundary Hamiltonian entries are not chosen by DFT calculations or experimental values they will not be accurate. However, as long as the boundary coefficients satisfy Assumption \ref{assump:decay}, the limit of the density of states $\D[H_{r,j}(0)]$ will be independent of the choice of boundary terms.
\end{remark}

\begin{remark}
   For the sake of convenience, we have chosen a circular shape for the
   approximating domains. Weaker requirements can be readily formulated, e.g.,
   domains $\widetilde{\Omega}$ should {\em contain} balls centered at the origin with radii
   growing to infinity, while at the same time keeping a suitable bound on the
   surface area to volume ratio.
\end{remark}

\begin{remark}
\label{remark:regularity}
The Riesz-Markov-Kakutani Representation Theorem states that the dual space of the continuous compact functions are the Radon measures. Since all our density of states and local density of states operators are continuous linear functionals over the space of compact continuous functions, they are all Radon measures.
\end{remark}

\begin{remark}
This methodology can easily be extended to three or more incommensurate layers, but at the cost of multiple integrals, since one must integrate over all relative shifts between the layers. The local density of states can be easily analyzed for multiple layers without adding much to the cost.
\end{remark}

\section{Numerical Simulations}
\label{sec:numerics}
\subsection{Quadrature}
\label{subsec:disc}

To compute the integrals occuring in Theorem \ref{thm:thermal} numerically, we
can use the smoothness properties from Theorem \ref{thm:smooth}, which can be
strengthened further by assuming analyticity on $h_{\alpha\alpha'}$.


\begin{thm}
\label{thm:smooth_rate}
Assume $h_{\alpha\alpha'}$ is analytic and satisfies Assumption \ref{assump:decay}. Let
\begin{equation*}
\DiscSpace_j = \left\{ A_j \begin{pmatrix}i_1/\DiscSize\\i_2/\DiscSize \end{pmatrix} : 0 \leq i_1,i_2 < \DiscSize \right\}
\end{equation*}
be the uniform discretization sample points. Then we have
\begin{align*}
   & \biggl|\frac{|\Gamma_{P_j}|}{\DiscSize^2}\sum_{b \in \DiscSpace_{P_j}}\sum_{\alpha \in \mathcal{A}_j} \D_\alpha[\M](b, g)-\sum_{\alpha \in \mathcal{A}_j} \int_{\Gamma_{P_j}} \D_\alpha[\M](b,g)db\biggl| \\
& \hspace{6cm}
      \leq C\tilde{d}^{-1} \sup_{z :~ d(z,\Msup) < \tilde{d}}|g(z)| e^{-\gamma''\tilde{d}\DiscSize}
\end{align*}
for some $\gamma'' > 0$.
\end{thm}

\begin{remark}
   In practice, $h_{\alpha\alpha’}$ has a finite cut-off and hence
   cannot be analytic. However, we can think of it as an approximation to an
   “exact” analytic $\bar{h}_{\alpha\alpha’}$. Preasymptotically, it is therefore
   useful to treat $h_{\alpha\alpha'}$ as if it were
   itself analytic.
\end{remark}

\subsection{Kernel Polynomial Method Approximation}
\label{subsec:approx}

A complete eigensolve on $H_{r,j}(b)$ for each quadrature point $b$
is computationally expensive, with scaling $O(r^{6})$.
Instead we use a Chebyshev Kernel Polynomial Method (KPM) to
compute the density of states \cite{kernel_poly}. This method scales as
$O(r^2)$, where the constant depends on the desired accuracy. It yields the
density of states operator as a smooth function from which multiple observables
can then be computed.

\begin{lemma}
\label{lemma:convolute}
Assume that $H$ satisfies Assumptions~\ref{assump:decay} and~\ref{assump:main} and that $f \in C(\mathbb{R} \times\mathbb{R}; \mathbb{C})$ and $g \in \Lambda$; then
\begin{equation*}
\int \D[\M]\big(f(\varepsilon,\cdot)\big)g(\varepsilon)d\varepsilon = \D[\M]\biggl(\int f(\varepsilon,\cdot)g(\varepsilon)d\varepsilon\biggr).
\end{equation*}

\end{lemma}
\begin{proof}
This result follows immediately from Remark \ref{remark:regularity} and Fubini's Theorem.
\end{proof}

We note that $|\D[\M](g)| \leq \|g\|_\infty$, and hence
\begin{equation}
\label{e:approx}
\left|\D[\M]\biggr(\int f(\varepsilon,\cdot)g(\varepsilon)d\varepsilon\biggl) - \D[\M](g)\right| \leq \left\| \int f(\varepsilon,\cdot)g(\varepsilon)d\varepsilon - g \right\|_\infty.
\end{equation}
 Note that this bound trivially extends from $\Lambda$ to $C(\Msup)$.  Moreover, if $f(\varepsilon,e) \approx \delta (\varepsilon - e)$, then the smooth function
\begin{equation*}
   D_f(\varepsilon) := \D[\M](f(\varepsilon,\cdot)) \approx \D[\M]
\end{equation*}
in the sense of Equation \ref{e:approx}. We now choose a convenient $f$.

Recall that the Chebyshev polynomials are a basis defined recursively by
\begin{equation}
\label{e:recursion_relation}
T_0(e) = 1, \qquad
T_1(e) = e, \qquad \text{and} \qquad
T_{n+1}(e) = 2eT_n(e)-T_{n-1}(e).
\end{equation}
The polynomials are orthogonal in the sense that
\begin{equation*}
\int_{-1}^1 \frac{1}{\pi\sqrt{1-e^2}} T_n(e)T_m(e)de = \frac{1+\delta_{0n}}{2} \delta_{nm}.
\end{equation*}
An approximation to the shifted delta function $\delta(e-\varepsilon),$ at
$\varepsilon\in(-1,1),$ is given by
\begin{equation*}
\hat \chi_p(\varepsilon,e) = \frac{1}{\pi\sqrt{1-\varepsilon^2}}\sum_{m \leq p} g_m^pT_m(\varepsilon)T_m(e),\qquad e, \varepsilon \in(-1,1),
\end{equation*}
where
\begin{equation*}
g_{m}^p = (2 - \delta_{m0})\frac{(p-m+1)\cos(\frac{\pi m}{p+1})+\sin(\frac{\pi m}{p+1})\arctan(\frac{\pi}{p+1})}{p+1}
\end{equation*}
are the so-called Jackson coefficients designed to remove the Gibbs
phenomenon \cite{kernel_poly}.

To approximate the density of states on the interval $\Msup= [-E[\M],E[\M]],$ we rescale
\begin{equation*}
   \chi_p(\varepsilon,e):=\eta\hat\chi_p(\eta\varepsilon,\eta e),
   \qquad e, \varepsilon\in(-1/\eta,1/\eta),
\end{equation*}
where $\eta$ is a positive constant selected so that $E[\M]\le 1/\eta.$


We approximate $\D[\M]$ by
\begin{equation*}
D_{\chi_p}(\varepsilon) = \nu \sum_{j=1}^2 \sum_{\alpha \in \mathcal{A}_j}\int_{\Gamma_{P_j}}\D_\alpha[\M](b,\chi_p(\varepsilon,\cdot))db,
\end{equation*}
and subsequently approximate the integrand
 $\D_\alpha[\M](b,\chi_p(\varepsilon,\cdot))$ by
\begin{equation}
\label{e:ldos}
\begin{split}
\D_\alpha[H_{r,j}(b)](\chi_p(\varepsilon,\cdot)) &= [\chi_p(\varepsilon,H_{r,j}(b))]_{0\alpha,0\alpha} \\
& = \frac{\eta}{\pi \sqrt{1-(\eta\varepsilon)^2}} \sum_{m\leq p} g_m^p T_m(\eta\varepsilon) \big[\eta T_m(H_{r,j}(b))\big]_{0\alpha,0\alpha}.
\end{split}
\end{equation}
 Note that for all $\varepsilon$, the calculation requires the same $[ T_m(\eta
 H_{r,j}(b))]_{0\alpha,0\alpha}$ coefficients, which is the
 core of our {\bf Algorithm A}.

\begin{alg}
   \medskip
   \begin{minipage}{15cm}

      {\bf Algorithm A: Approximate DoS} \\

      {\bf Step 1: } Choose quadrature parameter $\DiscSize \in \mathbb{N}$ and
      domain truncation radius $r >0$. For each $j \in \{1,2\}$ and $b \in \DiscSpace_{P_j}$ construct the matrix $H_{r,j}(b)$. \\

      {\bf Step 2: } Let $e_i \in \mathbb{R}^{|\Omega_r|}$ such that $[e_i]_j = \delta_{ij}$ is the $i^{\text{th}}$ coordinate vector. Using the recursion \eqref{e:recursion_relation} we compute,
      for $\alpha \in \mathcal{A}_j$,
      \begin{equation*}
      \begin{split}
      &\hspace{6mm} v_0 = e_{0\alpha} \\
      &\hspace{6mm} v_1 = \eta H_{r,j}(b)e_{0\alpha}\\
      &\hspace{6mm} \text{store: } [T_0(\eta H_{r,j}(b))]_{0\alpha,0\alpha} = e_{0\alpha} \cdot v_0 \text{ and } [T_1(\eta H_{r,j}(b)]_{0\alpha,0\alpha}  = e_{0\alpha} \cdot v_1\\
      &\hspace{6mm} \text{for loop: } 1 \leq m \leq p-1\\
      &\hspace{12mm} v_{m+1}= 2\eta H_{r,j}(b)v_m - v_{m-1}\\
      &\hspace{12mm}\text{store: } [T_{m+1}(\eta H_{r,j}(b))]_{0\alpha,0\alpha} = e_{0\alpha} \cdot v_{m+1}\\
      \end{split}
      \end{equation*}
      This yields the coefficients $[T_m(\eta H_{r,j}(b))]_{0\alpha,0\alpha}$
      for \eqref{e:ldos}. \\

   {\bf Step 3:} Compute the expression
      \begin{equation*}
      \D_{\alpha}[H_{r,j}(b)](\chi_p( \varepsilon,\cdot)) =\frac{\eta}{\pi\sqrt{1-(\eta \varepsilon)^2}}  \sum_{m \leq p} g_m^p T_m(\eta \varepsilon)[T_m(\eta H_{r,j}(b))]_{0\alpha,0\alpha}.
      \end{equation*}
      This yields a local density of states approximation, which is interesting in its own right. \\

      {\bf Step 4:} The total density of states approximation is obtained
      by evaluating
      \begin{equation*}
      D(\varepsilon) := \frac{\nu}{\DiscSize^2} \sum_{j=1}^2\sum_{\alpha \in \mathcal{A}_j}\sum_{b \in \DiscSpace_{P_j}} |\Gamma_{P_j}|
      \cdot\D_{\alpha}[ H_{r,j}(b)](\chi_p( \varepsilon,\cdot))
      \end{equation*}
      for all desired $\varepsilon$.

   \end{minipage}
   \medskip
\end{alg}

The approximation error for the output $D(\varepsilon)$ of Algorithm A
is estimated in the following result.

\begin{thm}
\label{thm:final}
   Suppose that $\M$ satisfies Assumptions~\ref{assump:main} and~\ref{assump:decay},
   then for
   $g \in \Lambda_{\tilde d}$,
\begin{equation*}
\begin{split}
\biggl|\D[\M](g) - \int D(\varepsilon)g(\varepsilon)d\varepsilon\biggr| & \leq \underbrace{C\tilde{d}^{-2} \sup_{d(z,\Msup) < \tilde{d}}|g(z)|e^{-\gamma \tilde{d}r}}_{\text{Truncation Error}} + \\
& \underbrace{ C\tilde{d}^{-1} \sup_{d(z,\Msup) < \tilde{d}}|g(z)| e^{-\gamma'\tilde{d}\DiscSize}}_{\text{Discretization Error}}+ \underbrace{C'\left\|g - \int \chi_p(\varepsilon,\cdot)g(\varepsilon)d\varepsilon \right\|_\infty }_{\text{Kernel Polynomial Method Error}}.
\end{split}
\end{equation*}
Here $\gamma,\gamma'>0$ are independent of the choice of $\tilde{d}$.
\end{thm}
\begin{proof}
The Truncation Error follows from Theorem \ref{thm:conv}, the Discretization Error from Theorem \ref{thm:smooth_rate}, and the Kernel Polynomial Error from  \eqref{e:approx}.
\end{proof}

\begin{remark}
   If we do not assume that $h_{\alpha\alpha'}$ is analytic and use $h_{\alpha\alpha'} \in C_0^n(\mathbb{R}^2)$ instead, the Truncation Error above is replaced with the standard periodic discretization error \cite[Theorem 1]{numeric_integration}, but the bound does not give the dependence of $\DiscSize$ on $\tilde{d}$.
\end{remark}

\subsection{Convergence Rates}
We briefly discuss a heuristic to choose the approximation
parameters  $p, \DiscSize \in \mathbb{N}$ and $r >0$. In
practice, one is interested in calculating the density of
states at a point or in calculating an observable
$\D[H](g)$ for $g \in \Lambda_{\tilde d}$.

For the first case, we note that $\chi_p$ acts similar to an approximation to
the identity of width proportional to $p^{-1}$ \cite{kernel_poly} with well
preserved regularity because of the Jackson coefficients. For analytic purposes,
we can consider $\chi_p(\varepsilon,e) \sim p^{-1} \phi((\varepsilon-e)/p)$ for
some analytic function $\phi$, $|\phi(x)| < e^{-c|x|}$ for $x \in \Msup$ and for
some $c >0$. An approximation of the density of states at a given energy point
$\varepsilon$ is given by $\D[H](\chi_p(\varepsilon,\cdot)) \sim
\D[H](p^{-1}\phi((\cdot - \varepsilon)/p)$. To approximate
$\D[H](\chi_p(\varepsilon,\cdot))$, we use Theorem \ref{thm:final} letting
$\tilde d \sim p^{-1}$ to see that the errors will be balanced if
\begin{equation}
   \label{e:Rrate}
   r \sim \DiscSize \sim p\log(p)
\end{equation}
Suppose the density of states is a function, i.e.,
\begin{equation*}
   \D[H](g) = \int \DoS(\epsilon) g(\epsilon)d\epsilon,
\end{equation*}
where $\DoS$ has Lipschitz constant $M$. Then we can estimate
\begin{equation*}
   |\DoS(\varepsilon) - \D[H](\chi_p(\varepsilon,\cdot)| \leq Mp^{-1}.
\end{equation*}
to obtain
\begin{equation*}
|D(\varepsilon) - \DoS(\varepsilon)| \leq C'(p e^{-\gamma' \frac{\DiscSize}{p}} + p^2 e^{-\gamma \frac{r}{p}} + Mp^{-1}).
\end{equation*}
If the constants in \eqref{e:Rrate} are chosen sufficiently small, we have
\begin{equation}
\label{e:pointwise_conv}
|D(\varepsilon) - \DoS(\varepsilon)| \leq (M+ C)p^{-1},
\end{equation}
where $C > 0$ is independent of smoothness properties of $\DoS$.

If the $\DoS$ is $C^2$ at a point $\varepsilon$ of interest, then we
may even expect
\begin{equation}
\label{e:pointwise_conv2}
|D(\varepsilon) - \DoS(\varepsilon)| \leq Cp^{-2},
\end{equation}
due to the fact that $\int xe^{-a x^2} dx = 0$ for any $a >0$.

For the second case, when the observable $g \in \Lambda$ is fixed
(no polynomial degree approximation parameter $p$), we have {\em in principle}
exponential decay of the error in $r$ and $\DiscSize$.
This seems to imply that it would be optimal to calculate the
observable directly using an eigensolve, thus avoiding the slower decay in
$p$. However the decay rate in $r$ is strongly coupled to the value of
$\tilde{d}$ from Theorem \ref{thm:conv}, which is fairly
small for interesting observables. Therefore, the involved matrices are
typically quite large, rendering direct eigensolves impractical.

\subsection{Numerical Results}
\label{subsec:num_results}
We test our approximation scheme using a tight-binding model for twisted bilayer graphene \cite{fangTB} with a relative twist angle of $6^\circ$. We fix an $\alpha \in \mathcal{A}_1$ and then verify numerically the following two results:

\begin{enumerate}
\item As predicted in Theorem \ref{thm:conv}, $\D_\alpha[H_{r,1}(b)](\chi_p(\varepsilon,\cdot)) \rightarrow \D_\alpha[\M](b,\chi_p(\varepsilon,\cdot))$ as $r \to \infty$
      with exponential rate: see Figure \ref{fig:r_conv}.
\item As predicted by Theorem \ref{thm:final} and \eqref{e:pointwise_conv2},
   $D \rightarrow \mathrm{DoS}$ pointwise as $p, r, \DiscSize \to \infty$,
   with quadratic rate: see Figure \ref{fig:l1_conv}.
\end{enumerate}

\begin{figure}[!htb]
\centering
\includegraphics[scale=.5]{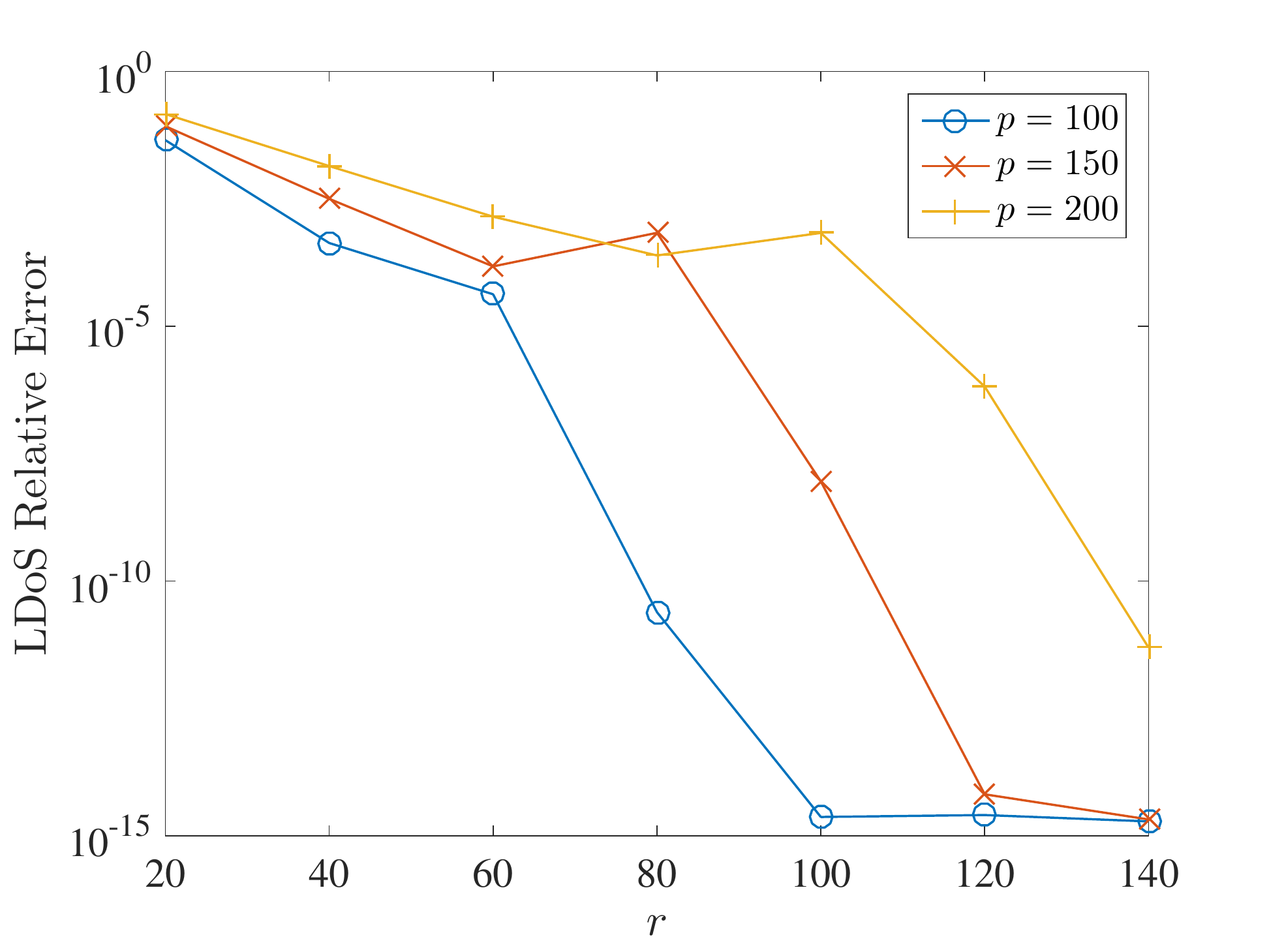}
\caption{ Relative error of $\D_\alpha[H_{r,1}(0)](\chi_{p}(0,\cdot))$ converging to $\D_\alpha[H](\chi_p(0,\cdot))$, for increasing values of $p$.}
\label{fig:r_conv}
\end{figure}

\begin{figure}[!htb]
\centering
\includegraphics[scale=.5]{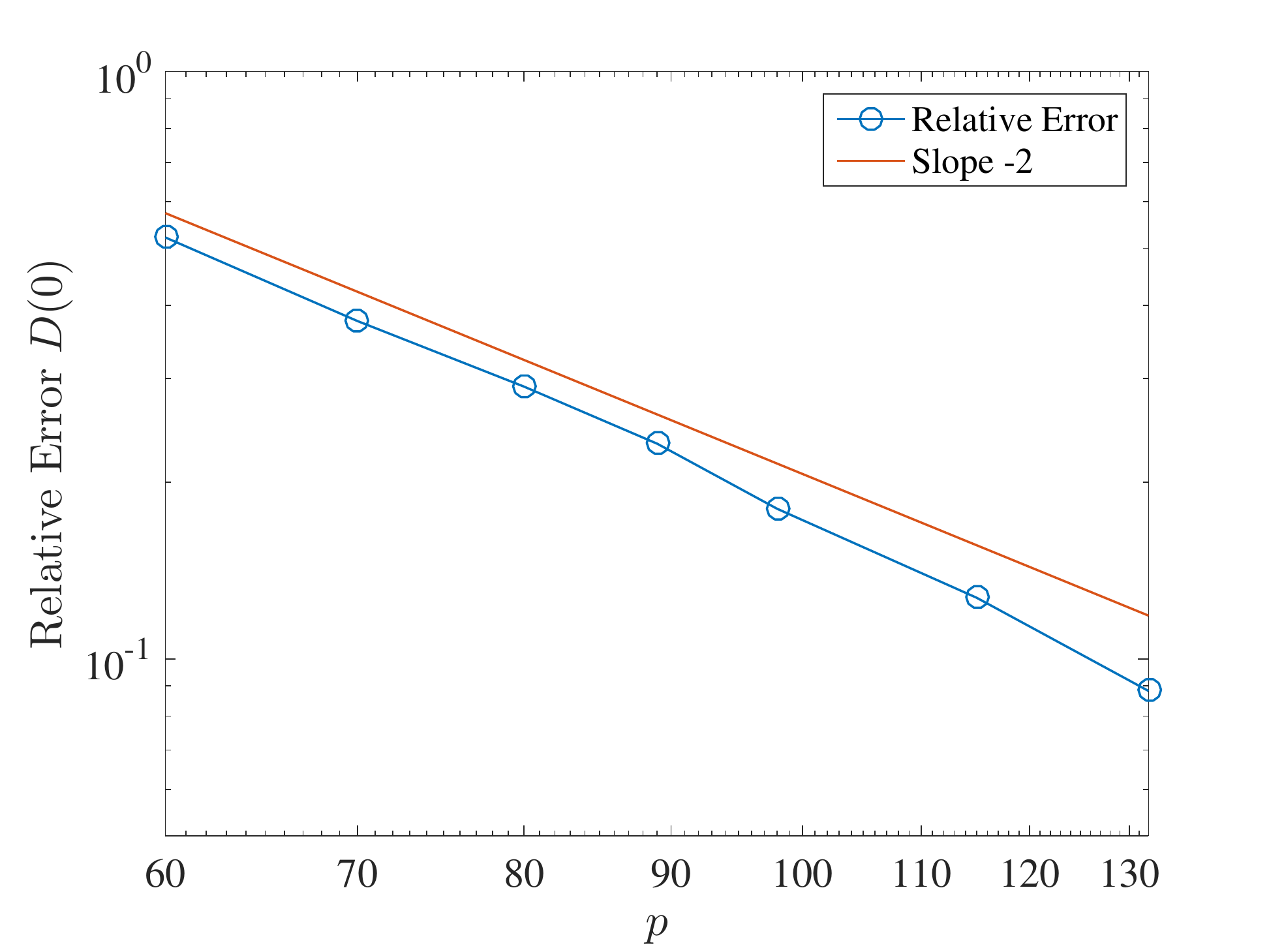}
\caption{ Relative error of $D(0) \rightarrow \text{DoS}(0)$ pointwise, where $r$ and $\DiscSize$ scale as in \eqref{e:Rrate}. The slope is $-1.98 \approx -2$, as predicted in \eqref{e:pointwise_conv2}.}
\label{fig:l1_conv}
\end{figure}

Furthermore, we demonstrate the practicality of Algorithm A by reproducing
twisted bilayer effects in the density of states of two stacked graphene sheets
with a relative twist of $6^\circ$ as predicted in \cite{fangTB} (See Figure
\ref{fig:dos}). We included the DoS for monolayer graphene for comparison.
The conical region near the $-.6$ energy region is called the Dirac cone. When
the two layers interact, the curve splits near the cone tip (the Dirac point)
forming two Van Hove Singularities on either side of the tip.
In practice the VHS needs higher resolution. We will explore how to achieve high resolutions in a future work.

\begin{figure}[!htb]
\centering
\includegraphics[scale=.5]{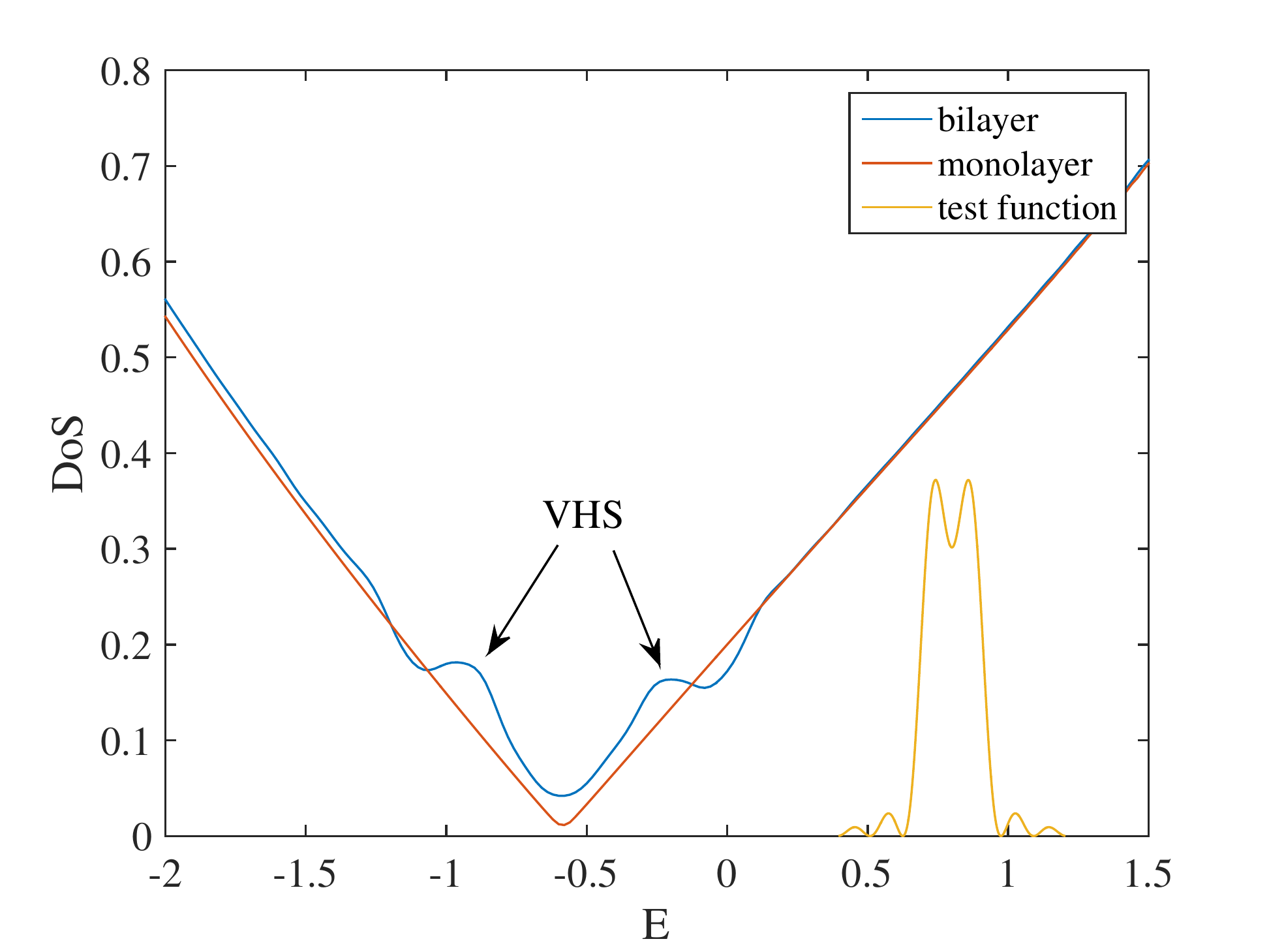}
\caption{ Approximation of the DoS with $r = 180$, $p = 700$, and $\DiscSize = 4$. We can see Van Hove Singularities (VHS) forming near the Dirac Point, agreeing with theoretical results \cite{fangTB}. We include the test function, which is to scale in the E-axis, but not in the DoS-axis.}
\label{fig:dos}
\end{figure}

\section{Proofs}
\label{sec:thm}
To attain bounds on the density of states objects, we will use resolvent bounds
as introduced in \cite{ChenOrtnerTB}. We denote $\C$ a contour around $\Msup$,
which contains the spectrum.  We can write for $\tilde \Omega \subset \Omega$
finite, $\tilde H \in M_{|\tilde \Omega|}(\mathbb{C})$, $k \in \Omega$, and $g$
analytic
\begin{equation*}
[g(H_{\tilde \Omega})]_{kk} =\frac{1}{2\pi i} \oint_\C g(z)[(z-\tilde H)^{-1}]_{kk}dz.
\end{equation*}
We will then rely on decay estimates for $[(z-\tilde H)^{-1}]_{kk}$ as $\tilde \Omega \uparrow \Omega$. We will vary our choice of $\C$ to tune the error bounds.

\subsection{Proof of Theorem \ref{thm:equidistribution} }
Although this result is conceptually close to the equidistribution theorem
\cite{Zorzi2015}, our specific statement of the result seems to be unavailable.
Hence we prefer to give a complete proof. Without loss of generality, we let $j =
1$ and hence $P_j = 2$. Then we wish to show for $g \in
C_{\text{per}}(\Gamma_{2})$, we have
\begin{equation*}
   \frac{1}{\# \R_1 \cap B_r} \sum_{\ell \in \R_1 \cap B_r} g(\ell) \rightarrow \frac{1}{|\Gamma_{2}|}\int_{\Gamma_{2}} g(b)db.
\end{equation*}
Upon transforming coordinates we may assume without loss of generality that $A_1 =
\text{Id}$. Hence for some matrix $A$ dependent on the original coordinates and
$V_r = |A B_r|$ we get
\begin{equation}
\label{e:ergodic}
   \frac{1}{V_r} \sum_{n \in \mathbb{Z}^2 \bigcap AB_r} g(n) \rightarrow \frac{1}{|\Gamma_2|}\int_{\Gamma_2} g(x)dx.
\end{equation}

Since $C_{\text{per}}^\infty(\Gamma_2)$ is dense in $C_{\text{per}}(\Gamma_2)$,
we assume $g \in C_{\text{per}}^\infty(\Gamma_2)$. On expanding $g$ into Fourier
modes, it suffices to show \eqref{e:ergodic} for an arbitrary fourier mode $g(x) =
e^{2\pi i m\cdot A_2^{-1}x}$ where $m \in \mathbb{Z}^2$.

If $m = (0, 0)$, then the left-hand side of \eqref{e:ergodic} converges
to $1$, which is the value of the right-hand side.

For $m \neq (0, 0)$, the left-hand side of \eqref{e:ergodic} vanishes, so
we need to prove that $\frac{1}{V_r} \sum_{n \in \mathbb{Z}^2 \bigcap AB_r} f(n) \to 0$
as $r \to \infty$. We first rewrite
\begin{equation*}
   \frac{1}{V_r} \sum_{n \in \mathbb{Z}^2 \bigcap AB_r} f(n)
   = \frac{1}{V_r} \sum_{n \in \mathbb{Z}^2 \bigcap AB_r} e^{2\pi i m^t A_2^{-1}n}
   = \frac{1}{V_r} \sum_{n \in \mathbb{Z}^2 \bigcap AB_r} e^{2\pi i a \cdot n},
\end{equation*}
where  $(a_1, a_2) = m^t A_2^{-1}$. If both $a_1$ and $a_2$ were rational, then
this would contradict Assumption \ref{assump:main}. Hence we assume, without
loss of generality, that $a_2 \notin \mathbb{Q}$.

Let $c > 0$ such that
\begin{equation*}
   n \in \mathbb{Z}^2 \bigcap AB_r \quad \Rightarrow  \quad
      n_1 \in [-cr, cr].
\end{equation*}
Moreover, for $n_1 \in [-cr, cr] \cap \Z$ let $f_1(n_1), f_2(n_2) \in \Z$ such that
$(n_1, n_2) \in \mathbb{Z}^2 \bigcap AB_r$ if and only if $f_1(n_1) \leq n_2 \leq f_2(n_2)$.

We can now compute
\begin{align*}
   \frac{1}{V_r} \sum_{n \in \mathbb{Z}^2 \bigcap AB_r} e^{2\pi i a \cdot n}
   &=
   \frac{1}{V_r} \sum_{n_1 \in [-cr, cr] \cap \Z} e^{2\pi i a_1 n_1}
      \sum_{n_2 = f_1(n_1)}^{f_2(n_2)} e^{2 \pi i a_2 n_2} \\
   &=
   \frac{1}{V_r} \sum_{n_1 \in [-cr, cr] \cap \Z} e^{2\pi i a_1 n_1}
      \, \frac{e^{2\pi i a_2 (f_1(n_1)+1)} - e^{2\pi i a_2 (f_2(n_1)+1)}}{1 - e^{2\pi i a_2}}.
\end{align*}
Since $a_2$ is irrational, $1-e^{2\pi i a_2} \neq 0$, hence we can estimate
\begin{equation*}
   \bigg| \frac{1}{V_r} \sum_{n \in \mathbb{Z}^2 \bigcap AB_r} e^{2\pi i a \cdot n} \bigg|
   \leq \frac{4 c r}{|1-e^{2\pi i a_2}| V_r} \leq C r^{-1},
\end{equation*}
which vanished in the limit $r \to \infty$, as required. This completes the proof
of Theorem \ref{eq:equidistribution}.


\subsection{Proof of Theorem \ref{thm:conv}}

Recall that
\begin{equation*}
\Lambda := \{ g \in C(\mathbb{R}) ~|~ g \text{ is analytic on }\Msup \}.
\end{equation*}
In particular, note that $\Lambda$ is dense in $C(\Msup)$, in the sense that for any $f \in C(\Msup)$ and $\epsilon > 0$, there exists $g \in \Lambda$ such that
\begin{equation*}
\| g|_{\Msup} - f\|_\infty < \epsilon.
\end{equation*}
This will be useful for extending the density of states operators from $\Lambda$ to $C(\Msup)$.

\begin{lemma}
\label{lemma:resolvent}
Suppose $\Mat \in M_n(\mathbb{C})$, and $y : \{1,2,\cdots, n\} \rightarrow \mathbb{R}^2$ such that
\begin{equation*}
|\Mat_{k\ell}| \leq C e^{-\tilde{\gamma} |y(k)-y(\ell)|}
\end{equation*}
for some $\tilde{\gamma} > 0$. Let $N \in \mathbb{N}$, $r' > 0$ and suppose that
for all $x \in \mathbb{R}^2$ $|\#\{ y(j) : y(j) \in B_{r'}(x)\}| < N$.  Then
there exists $\gamma >0$ such that, for all $z \in \mathbb{C}$, ${\rm dist}(z,\Msup) \geq \tilde{d}$,
\begin{equation*}
\Big| [(z-\Mat)^{-1}]_{k\ell} \Big| \leq C'\tilde{d}^{-1} e^{-\gamma \tilde{d} |y(k)-y(\ell)|}
\end{equation*}
Here $C'$ and $\gamma$ are dependent on $\tilde{\gamma}, N,r',$ and $C$.

\end{lemma}

\begin{proof}
   This is a version of Lemma 2.2 from \cite{ChenOrtnerTB}.
\end{proof}

In particular, the previous lemma applies to the matrices $H_{r,j}(b)$. To apply
it we will set $y = \Rinv$ where in the following we define
\begin{equation*}
   \Rinv : \Omega \rightarrow \mathbb{R}^2, \qquad
   \Rinv( R\alpha) = R.
\end{equation*}
For the next lemma, recall the definition of $H_{r',j}(b)$ from
\eqref{eq:H_rj_b}.

\begin{lemma} \label{lemma:resolvent2}
   Suppose that $H$ satisfies Assumptions~\ref{assump:decay} and~\ref{assump:main}.
   Let $\tilde{\Omega} \subset \Omega$ be a set of indices and
    $\tilde H_j(b)$ be the matrix defined over $\tilde \Omega$ with shift $b$ relative to sheet $j$, that is,
   \begin{equation*}
   [\tilde H_j(b)]_{R\alpha,R'\alpha'} = h_{\alpha\alpha'}\bigl(b(\delta_{\alpha \in \mathcal{A}_{P_j}} - \delta_{\alpha' \in \mathcal{A}_{P_j}})+R-R'\bigr).
   \end{equation*}

   Suppose that $r' > 0$ such that $\Omega_{r'} \subset \tilde \Omega$
   and $\tilde{d} > 0$ such that $d(z,S[H]) > \tilde{d}$, then
   \begin{equation*}
      \begin{split}
      &\biggl|\big[(z-\tilde H_j(b))^{-1}\big]_{k\ell}
         - \big[(z- H_{r',j}(b))^{-1} \big]_{k\ell}\biggr|\\
       &\qquad\qquad\leq C \tilde{d}^{-2} \min\big\{e^{- \gamma \tilde{d} |\Rinv(k)-\Rinv(\ell)|} , r'e^{-\gamma \tilde{d} \min\{r'-|\Rinv(k)|,r'-|\Rinv(\ell)|\}} \big\},
      \end{split}
   \end{equation*}
   where $C$ and $\gamma$ are independent of $\tilde{\Omega}$ and $r'$ (See Figure \ref{fig:resolvent2}).
\end{lemma}
\begin{proof}
   We define the matrix $\tilde H_{j}^{r'}(b) \in
   M_{|\tilde \Omega|}(\mathbb{C})$ such that
\[ [\tilde H_{j}^{r'}(b)]_{k\ell} =
  \begin{cases}
    H_{r',j}(b)       & \quad \text{if } k,\ell \in \Omega_{r'}\\
   0			& \quad \text{otherwise}\\
  \end{cases}.
\]

We write $\tilde H_{j}(b) =   \tilde H_{j}^{r'}(b) + \bigl(\tilde H_{j}(b)-\tilde H_{j}^{r'}(b)\bigr)$, and
\begin{equation*}
[(z-\tilde H_{j}(b))^{-1}]_{k\ell} =[ \bigl(z-\tilde H_{j}^{r'}(b) - (\tilde H_{j}(b) -\tilde H_{j}^{r'}(b))\bigr)^{-1}]_{k\ell}.
\end{equation*}
Thus, after defining
\begin{equation*}
B(\lambda) = z - \tilde H_{j}^{r'}(b) - \tilde \lambda(\tilde H_{j}(b)-\tilde H_{j}^{r'}(b)),
\end{equation*}
and
\begin{equation*}
f(\lambda) = [ B(\lambda)^{-1}]_{k\ell}
\end{equation*}
we need to estimate $f(1) - f(0)$. Differentiating with respect to $\lambda$
yields
\begin{equation*}
\begin{split}
f'(\lambda) &=[ B(\lambda)^{-1} (\tilde H_{j}(b)-\tilde H_{j}^{r'}(b))B(\lambda)^{-1}]_{k\ell}\\
& = \sum_{t,s \in\tilde  \Omega}   [B(\lambda)^{-1}]_{kt} [(\tilde H_{j}(b)-\tilde H_{j}^{r'}(b))]_{t s}[B(\lambda)^{-1}]_{s\ell}.
\end{split}
\end{equation*}
Now $[\tilde H_{j}(b) - \tilde H_{j}^{r'}(b))]_{t s}$ is only nonzero if $t$ or $s \notin \Omega_{r'}$. We use the definition
\begin{equation*}
\tilde \Omega\setminus \Omega_{r'} := \{ x : x \in \tilde \Omega, x \notin \Omega_{r'}\}.
\end{equation*}
From Lemma \ref{lemma:resolvent}, we have
\begin{equation*}
|\tilde{H}(\lambda)^{-1}|_{st} \leq C\tilde{d}^{-1} e^{-\gamma\tilde{d}|\mathfrak{R}(s)-\mathfrak{R}(t)|}.
\end{equation*}
Therefore, we obtain the bound
\begin{equation*}
\begin{split}
|f'(\lambda)| & \leq  \sum_{t \in\tilde \Omega}\sum_{s \in \tilde \Omega\setminus\Omega_{r'}}\bigl| [B(\lambda)^{-1}]_{kt}[\tilde H_{j}(b)-\tilde H_{j}^{r'}(b)]_{t s}[B(\lambda)^{-1}]_{s\ell}\bigl| \\
& \qquad \qquad +\sum_{s\in\tilde \Omega}\sum_{t \in\tilde \Omega\setminus\Omega_{r'}}\bigl| [B(\lambda)^{-1}]_{kt}[\tilde H_{j}(b)-\tilde H_{j}^{r'}(b)]_{t s}[B(\lambda)^{-1}]_{s\ell}\bigl| \\
& \leq C\tilde{d}^{-2}\sum_{t \in \tilde \Omega}\sum_{s \in \tilde \Omega\setminus\Omega_{r'}} e^{-\gamma\tilde{d}( |\Rinv(k)-\Rinv(t)| + |\Rinv(t) - \Rinv(s)| + |\Rinv(s)-\Rinv(\ell)|)}  \\
& \qquad \qquad +C\tilde{d}^{-2}\sum_{s \in\tilde  \Omega}\sum_{t \in \tilde \Omega\setminus\Omega_{r'}} e^{-\gamma\tilde{d}( |\Rinv(k)-\Rinv(t)| + |\Rinv(t) - \Rinv(s)| + |\Rinv(s)-\Rinv(\ell)|)}  \\
& \leq C'\tilde{d}^{-2} \min\{e^{- \gamma \tilde{d} |\Rinv(k)-\Rinv(\ell)|} , r'e^{-\gamma \tilde{d} \min\{r'-|\Rinv(k)|,r'-|\Rinv(\ell)|\}} \}.
\end{split}
\end{equation*}
Hence, we conclude that
\begin{equation*}
\begin{split}
\bigl|[(z- \tilde H_{j}(b))^{-1}]_{k\ell} - [(z-H_{r',j}(b))^{-1}]_{k\ell}\bigr| &\leq |f(1)-f(0)| \leq \int_0^1 |f'(\lambda)|d\lambda \\
& \hspace{-4cm}\leq C'\tilde{d}^{-2} \min\{e^{- \gamma \tilde{d} |\Rinv(k)-\Rinv(\ell)|} , r'e^{-\gamma \tilde{d} \min\{|r'-|\Rinv(k)|,r'-|\Rinv(\ell)|\}} \}. \qedhere
\end{split}
\end{equation*}
\end{proof}

\begin{figure}[ht]
\centering
   \includegraphics[width=0.6\textwidth]{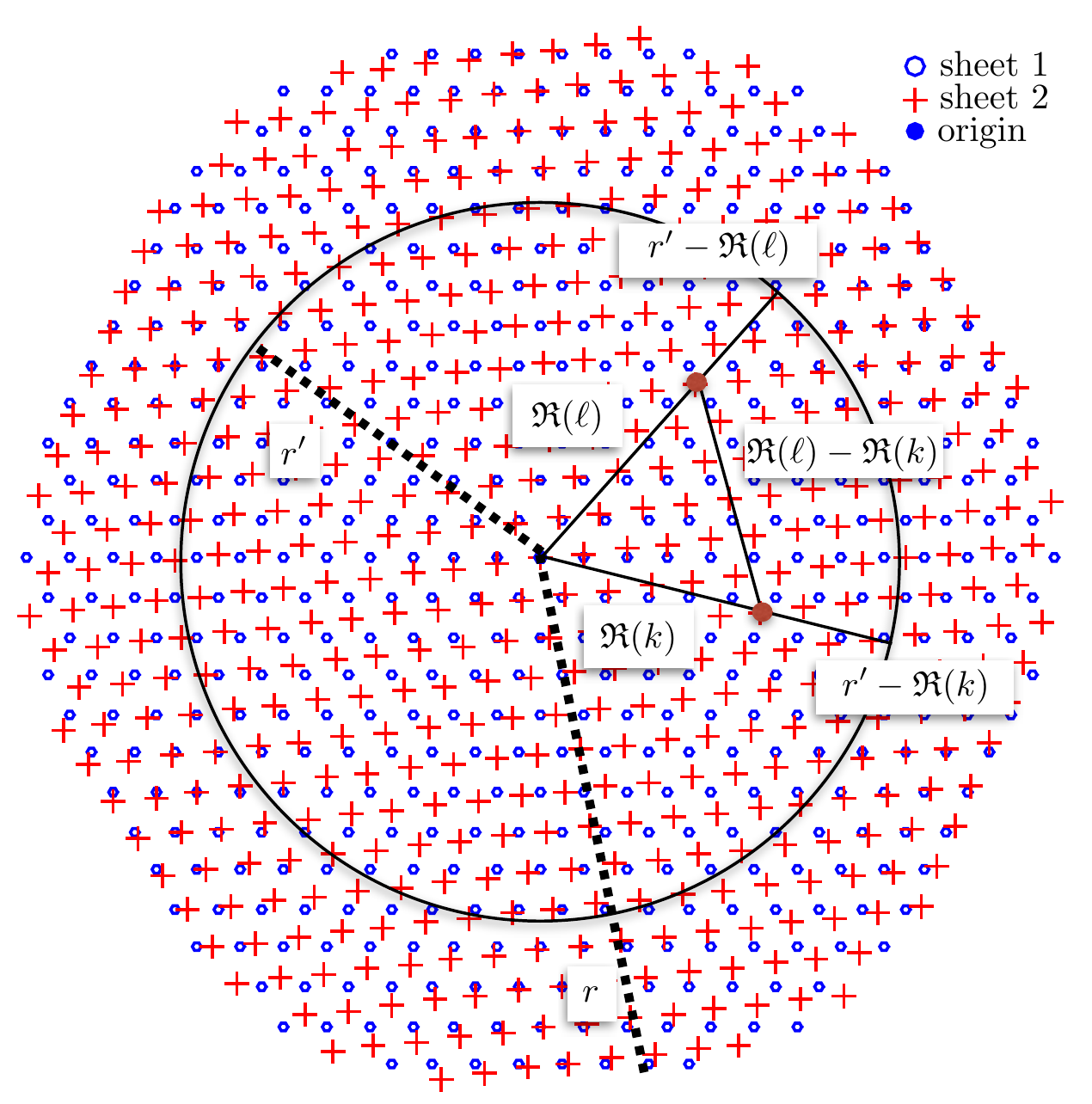}
\caption{For given sites $\ell$ and $k$, we plot the relevant distances in solid lines and system radii in dotted lines for considering resolvent error in Lemma \ref{lemma:resolvent2}. }
\label{fig:resolvent2}
\end{figure}

Lemma \ref{lemma:resolvent2} shows that the resolvent difference is bounded by the site distances from the edge of the first cut-off region (the circle with radius $r'$) and the distance between the two sites. This is consistent with  Lemma \ref{lemma:resolvent}.

Let $\mathcal{C}$ be a contour around $\Msup$ such that $\tilde{d}/2 < d(\mathcal{C},\Msup) < \tilde{d}$. By Lemma \ref{lemma:resolvent2}, we have for $g \in \Lambda_{\tilde{d}}$ that
\begin{equation*}
\begin{split}
&\hspace{-1cm} |\D_{\alpha}[H_{r,j}(b)](g) - \D_{\alpha}[H_{r',j}(b)](g)| \\
&=\biggl| \frac{1}{2\pi i}\oint_\C g(z)\biggr([(z-H_{r,j}(b))^{-1}]_{0\alpha,0\alpha} - [(z- H_{r',j}(b))^{-1}]_{0\alpha,0\alpha}\biggl)\biggr|\\
& \leq C' \tilde{d}^{-2}r' \sup_{z \in \C}|g(z)| e^{-\gamma\tilde{d} r'}.
\end{split}
\end{equation*}
Hence $\{\D_\alpha[H_{r_n,j}(b)]\}_n$ is a Cauchy sequence for $r_n \rightarrow
\infty$, which therefore has some limit $\D_\alpha[\M](b,g)$. $\D_\alpha[\M]$ is
linear in $g$, since each element of the Cauchy sequence is linear. Further, we
have the error bound
\begin{equation*}
|\D_\alpha[\M](b,g) - \D_\alpha[H_{r,j}(b)](g)| \leq C' \tilde{d}^{-2}r\sup_{z \in \C}|g(z)| e^{-\gamma\tilde{d} r}.
\end{equation*}

Since the linear functional $\D_\alpha[H_{r,j}(b)]$ is bounded  by $\|\D_\alpha[H_{r,j}(b)]\| \leq 1$ we also obtain that $\D_\alpha[\M](b,\cdot)$ is a bounded linear functional, and so has a unique extension to a bounded linear functional on the space $C(\Msup)$.

This completes the proof of Theorem \ref{thm:conv}.

\subsection{Proof of Theorem \ref{thm:smooth}}

\begin{lemma}
\label{lemma:smooth}
Suppose $h_{\alpha\alpha'} \in C^n(\mathbb{R}^2)$ for $n \in \mathbb{N}\cup\{\infty\}$ and $\partial_{b_1}^m\partial_{b_2}^{m'}h_{\alpha\alpha'}$ is uniformly continuous for $m+m' \leq n$. We further assume the decay estimate
\begin{equation} \label{eq:lemma_smooth:exp_decay_Dh}
|\partial_{b_m}\partial_{b_{m'}} h_{\alpha\alpha'}(r)| \leq Ce^{-\gamma'r}.
\end{equation}
Then for $k = 0\alpha$, we have $b \mapsto [(z-H_{r,j}(b))^{-1}]_{kk} \in C^n_{\rm per}(\Gamma_{P_j})$, and we have the limit
\begin{equation*}
   b \mapsto \lim_{r\rightarrow\infty} [(z-H_{r,j}(b))^{-1}]_{kk} \in C_{\rm per}^n(\Gamma_{P_j})
\end{equation*}
for $d(\{z\},\Msup) > 0$. Furthermore, for all $b \in \mathbb{R}^2$, $z
\mapsto [(z-H_{r,j}(b))^{-1}]_{kk}$ is analytic in $\mathbb{C}\setminus
\Msup$.
\end{lemma}

\begin{proof}
   We will only consider the derivative $\partial_{b_1}$; the
   treatment of higher (and lower) order derivatives follow the same
   line of argument, but are more cumbersome.
Let $k = 0\alpha$ for some $\alpha \in \mathcal{A}_j$, then
\begin{equation*}
\begin{split}
\partial_{b_1} [(z-H_{r,j}(b))^{-1}]_{kk} &= \partial_{b_1}  [(z-H_{r,j}(b))^{-1}]_{kk}\\
&=\sum_{s,\ell \in \Omega_r}[(z-H_{r,j}(b))^{-1}]_{ks}[\partial_{b_1} H_{r,j}(b)]_{s\ell}[(z-H_{r,j}(b))^{-1}]_{\ell k}.
\end{split}
\end{equation*}
Lemma \ref{lemma:resolvent2} implies that, for  $r > r' > 0$,
\begin{equation*}
\begin{split}
R(r,r',k,s) &: =  \biggr|[(z-H_{r,j}(b))^{-1}]_{ks}-[(z-H_{r',j}(b))^{-1}]_{ks}\biggl|\\
& \leq C\min\{e^{-\gamma |\Rinv(k)-\Rinv(s)|},r'e^{-\gamma \min\{ r'-|\Rinv(k)|,r'-|\Rinv(s)|\}}\},
\end{split}
\end{equation*}
where $C$ and $\gamma$ are independent of $r$.
Note also that, for $s,\ell \in \Omega_{r'}$, we have
\begin{equation*}
\partial_{b_1}  [H_{r,j}(b)]_{s\ell} = \partial_{b_1}  [H_{r',j}(b)]_{s\ell}.
\end{equation*}
Recalling that $\Rinv(k) = \Rinv(0\alpha) = 0$,
and employing \eqref{eq:lemma_smooth:exp_decay_Dh}, we estimate
 \begin{equation*}
 \begin{split}
 &\bigl| \partial_{b_1}  [(z-H_{r,j}(b))^{-1}]_{kk} - \partial_{b_1}  [(z-H_{r',j}(b))^{-1}]_{kk} \bigr| \\
 &\leq C  \biggl(\sum_{s,\ell \in \Omega_{r'}} \bigl( R(r,r',k,s)e^{-\gamma|\Rinv(\ell)-\Rinv(k)|} + R(r,r',\ell,k)e^{-\gamma|\Rinv(s)-\Rinv(k)|}\bigr) |\partial_{b_1}  [H_{r,j}(b)]_{s\ell}|\\
& \qquad \qquad + \sum_{s \in \Omega_r, \ell \in \Omega_r\setminus\Omega_{r'}} \biggl|[(z-H_{r,j}(b))^{-1}]_{ks}[\partial_{b_1}  H_{r,j}(b)]_{s\ell}[(z-H_{r,j}(b))^{-1}]_{\ell k}\biggr|\biggr) \\
& \leq C'\biggl(\sum_{s,\ell \in \Omega_{r'}} \bigl( R(r,r',k,s)e^{-\gamma|\Rinv(\ell)|} + R(r,r',\ell,k)e^{-\gamma|\Rinv(s)|}\bigr) e^{-\gamma'|\Rinv(s)-\Rinv(\ell)| } + r'e^{-\gamma r'}\biggr) \\
& \leq  C'' r'\sum_{s,\ell \in\Omega_{r'}} e^{-\gamma(r'-|\Rinv(s)|) - \gamma |\Rinv(\ell)| - \gamma' |\Rinv(s)-\Rinv(\ell)|} + C'r'e^{-\gamma r'} \\
& \leq C''' r'^2\sum_{s \in \Omega_{r'}} e^{-\gamma(r'-|\Rinv(s)|) - \min\{\gamma,\gamma'\} |\Rinv(s)|} + C'r'e^{-\gamma r'} \\
& \leq \tilde{C} e^{-\gamma'' r'},
 \end{split}
 \end{equation*}
for any choice of $\gamma'' < \min\{\gamma,\gamma'\}$, where $\tilde{C}$ depends on the choice of $\gamma''$.

Therefore, as $r_n \rightarrow \infty$, $[(z-H_{r_n,j}(b))^{-1}]_{kk}$ forms a Cauchy sequence, and in particular has a limit
\begin{equation*}
   L_1(b) := \lim_{r \uparrow \infty}
            \partial_{b_1} [(z-H_{r,j}(b))^{-1}]_{kk}.
\end{equation*}

Next, we define
\begin{equation*}
   L(b) := \lim_{r\uparrow\infty}[(z-H_{r,j}(b))^{-1}]_{kk}.
\end{equation*}
We need to show that $\partial_{b_1}  L$ exists and satisfies
\begin{equation*}
\partial_{b_1}  L = L_1.
\end{equation*}
We denote
\begin{equation*}
\Resolvent(b) = [(z-H_{r,j}(b))^{-1}]_{kk}.
\end{equation*}
Since $\partial_{b_1} h_{\alpha\alpha'}$ is uniformly continuous there
exists a modulus of continuity $\omega$ such that $|\partial_{b_1}h(b) - \partial_{b_1}h(b')| \leq \omega(|b-b'|)$. We then observe that, for $\epsilon > 0$ and $e_1 = (1, 0)$,
\begin{equation*}
\begin{split}
\frac{1}{\epsilon}\biggr(\Resolvent(b+\epsilon e_1)  - \Resolvent(b) \biggl) &=  [(z-H_{r,j}(b))^{-1} \bigr( \partial_{b_1} H_{r,j}(b) + O(\omega(\epsilon))\bigl) (z-H_{r,j}(b))^{-1}]_{kk} \\
&= \partial_{b_1} \Resolvent(b) + O(\omega(\epsilon)).
\end{split}
\end{equation*}
Here $O(\omega(\epsilon))$ is independent of $r$. Letting $r \rightarrow \infty$, we have
\begin{equation*}
\frac{1}{\epsilon}\biggr([L(b+\epsilon e_1)  - L(b) \biggl) = L_1(b) + O(\omega(\epsilon)).
\end{equation*}
Letting $\epsilon \rightarrow 0$ shows that $L \in C_{\text{per}}^{(1,0)}(\Gamma_j)$ and
$\partial_{b_1}  L = L_1$,
which is the desired result.

Continuity with respect to $b$ follows the same argument. Analyticity
with respect to $z$ follows from Section 5.2 of \cite{kato1995perturbation}.
\end{proof}

Theorem \ref{thm:smooth} follows immediately from Lemma \ref{lemma:smooth}.

\subsection{Proof of Theorem \ref{thm:thermal} }

Without loss of generality, let $j = 1$. Fix $g \in \Lambda, r> 0$ and $\eta < 1$. Then we have
\begin{equation*}
\begin{split}
\D[H_{r,1}(0)](g) &= \frac{1}{|\Omega_r|}\sum_{k \in \Omega_r} \D_k[H_{r,1}(0)](g) \\
& = \frac{1}{|\Omega_r|}\biggl(\sum_{k \in \Omega_r \setminus \Omega_{\eta r}} \D_k[H_{r,1}(0)](g) + \sum_{k \in\Omega_{\eta r}} \D_k[H_{r,1}(0)](g)\biggr).
\end{split}
\end{equation*}

We define $\Ainv : \Omega \rightarrow \mathcal{A}_1 \cup \mathcal{A}_2$ such that $\Ainv(R\alpha) = \alpha$. By Lemma \ref{lemma:resolvent2}, we have for $k = R\alpha \in \Omega_{\eta r}$ and $\alpha \in \mathcal{A}_{j}$ that
\begin{equation*}
|\D_k[H_{r,1}(0)](g) -\D_\alpha[\M](\mod_{P_j}\circ\Rinv(k),g)| \leq C \sup_{z \in \C} |g(z)| e^{-\gamma r(1-\eta)}.
\end{equation*}
The site $k$ is at least a distance $r(1-\eta)$ from the boundary of $\Omega_r$.

Consider the distribution
\begin{equation*}
\D[\M](g) = \nu \sum_{j=1}^2\sum_{\alpha \in \mathcal{A}_j} \int_{\Gamma_{P_j}} D_\alpha[\M](b,g)db.
\end{equation*}
Since the integrand is continuous with respect to $b$ (see Theorem \ref{thm:smooth})
the integration is well-defined. We now estimate
\begin{equation*}
\begin{split}
| \D[\M](g) - \D[H_{r,1}(0)](g)| &\leq \biggl|\frac{1}{|\Omega_r|}\sum_{ k \in \Omega_r\setminus\Omega_{\eta r}} \D_k[H_{r,1}(0)](g)\biggr| \\
& \hspace{-1cm}  +\biggl| \D[\M](g) - \frac{1}{|\Omega_{\eta r}|} \sum_{j=1}^2\sum_{R\alpha \in \Omega_{\eta r}: \alpha \in \mathcal{A}_j} \D_\alpha[\M](\mod_{P_j}(R),g)  \biggr| \\
& \hspace{-1cm} + \biggl| \frac{1}{|\Omega_{\eta r}|}\sum_{j=1}^2\sum_{R\alpha \in \Omega_{\eta r}: \alpha \in \mathcal{A}_j} \D_\alpha[\M](\mod_{P_j}(R),g)\\
&  \hspace{-1cm} \hspace{5cm} - \frac{1}{|\Omega_{\eta r}|}\sum_{k \in \Omega_{\eta r}} \D_k[H_{r,1}(0)](g) \biggr| \\
&\hspace{-1cm} +
   \bigg(1 - \frac{|\Omega_{\eta r}|}{|\Omega_r|}\bigg)
   \frac{1}{|\Omega_{\eta r}|} \bigg|\sum_{k \in \Omega_{\eta r}} \D_k[H_{r,1}(0)](g) \bigg|.
\end{split}
\end{equation*}
The first and fourth terms are easily seen to be bounded by $O(1-\eta^2)$. By Theorem \ref{thm:equidistribution}, the second term converges to $0$ as $ r \rightarrow \infty$.
Finally, the third term can be estimated by
\begin{equation*}
\begin{split}
   \biggl| \frac{1}{|\Omega_{\eta r}|}
            \sum_{j=1}^2
            \sum_{R\alpha \in \Omega_{\eta r} : \alpha \in \mathcal{A}_j}
            \D_\alpha[\M](\mod_{P_j}(R),g)& \\
- \frac{1}{|\Omega_{\eta r}|}\sum_{R\alpha\in\Omega_{\eta r}} \D_{R\alpha}[H_{r,1}(0)](g) \biggr| &\leq C \sup_{z \in \C} |g(z)| e^{-\gamma r(1-\eta)}.
 \end{split}
\end{equation*}
Therefore if we choose a pair of sequences $(\eta_j),(r_j)$ such that $\eta_j \uparrow 1$, $r_j\uparrow \infty$, and $r_j(1-\eta_j) \rightarrow \infty$, we conclude that
\begin{equation*}
   \D[H_{r,1}(0)](g) \rightarrow \D[\M](g).
\end{equation*}
Since $\D[\M]$ is a bounded linear functional, it can be extended as before to be a bounded linear functional over $C(\Msup)$.

\subsection{Proof of Theorem \ref{thm:smooth_rate}}
We denote $\tilde z = (\tilde z_1, \tilde z_2) \in \mathbb{C}^2$. Let $z \in \mathbb{C}$. Then if $c>0$ is sufficiently small and $\Imaginary(\tilde z_1),\Imaginary(\tilde z_2) \in (-c,c)$, we have
\begin{equation*}
\|z-H_{r,j}(\tilde z)\|_2 > 0,
\end{equation*}
and hence $\oint_\C g(z)[(z-H_{r,j}(\tilde z))^{-1}]_{0\alpha,0\alpha}$ is analytic at $\tilde z$ satisfying $\Imaginary(\tilde z_1),\Imaginary(\tilde z_2) \in (-c,c)$. We pick a contour $\C$ enclosing $\Msup$ such that $\tilde{d}/2 < d(\mathcal{C},\Msup) < \tilde{d}$ and then chose $c>0$ small enough, but keeping $c \sim \tilde d$. Since $\int_\C g(z) [(z-H_{r,j}(\tilde z))^{-1}]_{0\alpha,0\alpha}dz$ is analytic with respect to $\tilde z$, we can apply Theorem 2 of \cite{numeric_integration} to deduce
\begin{equation*}
\bigl|\oint_\C g(z)[(z-H_{r,j}(\tilde z))^{-1}]_{0\alpha,0\alpha}\bigr|<C\tilde{d}^{-1}\sup_{z:~d(z,\Msup)<\tilde d}|g(z)| e^{-\gamma''\tilde{d}\DiscSize}
\end{equation*}
for some $C>0$ independent of $r$. The result follows.

\section{Conclusion} \label{sec:conclusion}
The main result of this work, Theorem~\ref{thm:thermal},
is a representation formula for the thermodynamic limit of the
electronic structure of incommensurate layered heterostructures.
The result is reminiscent of Bellisard's noncommutative Brillouin Zone for
aperiodic solids \cite{Bellissard2002}, replacing on-site randomness with a
number-theoretic equidistribution theorem.

Crucially, our representation formula lends itself to numerical approximation.
In \S~\ref{sec:numerics} we formulate, and analyze at a heuristic level, an
efficient kernel polynomial method to approximately compute the density of
states in twisted bilayer graphene. This preliminary exploration  provides not
only quantitative confirmation of our analytical results, but also demonstrates
the utility of our approach for applications to real material models.

\section*{Acknowledgement}
The authors would like to thank Stephen Carr and Paul Cazeaux for helpful comments on
the theme of this paper.


\begin{thebibliography}{10}

\bibitem{Bellissard2002}
J.~Bellissard.
\newblock {\em Dynamics of Dissipation}, chapter Coherent and Dissipative
  Transport in Aperiodic Solids: An Overview, pages 413--485.
\newblock Springer Berlin Heidelberg, Berlin, Heidelberg, 2002.

\bibitem{Castro2009}
A.~H. Castro~Neto, F.~Guinea, N.~M.~R. Peres, K.~S. Novoselov, and A.~K. Geim.
\newblock The electronic properties of graphene.
\newblock {\em Rev. Mod. Phys.}, 81:109--162, Jan 2009.

\bibitem{ChenOrtnerTB}
H.~{Chen} and C.~{Ortner}.
\newblock {QM/MM methods for crystalline defects. Part 1: Locality of the tight
  binding model}.
\newblock {\em ArXiv e-prints}, May 2015.

\bibitem{Ebnonnasir2014}
A.~Ebnonnasir, B.~Narayanan, S.~Kodambaka, and C.~V. Ciobanu.
\newblock Tunable {MoS2} bandgap in {MoS2}-graphene heterostructures.
\newblock {\em Applied Physics Letters}, 105(3), 2014.

\bibitem{fangTB}
S.~{Fang}, R.~{Kuate Defo}, S.~N. {Shirodkar}, S.~{Lieu}, G.~A. {Tritsaris},
  and E.~{Kaxiras}.
\newblock {Ab initio tight-binding Hamiltonian for transition metal
  dichalcogenides}.
\newblock {\em Phys. Rev. B}, 92(20):205108, Nov. 2015.

\bibitem{Huang:2006go}
C.~Huang, A.~Voter, and D.~Perez.
\newblock {The kernel polynomial method}.
\newblock {\em Rev. Mod. Phys.}, 78(1), Mar. 2006.

\bibitem{kato1995perturbation}
T.~Kato.
\newblock {\em Perturbation Theory for Linear Operators}.
\newblock Classics in Mathematics. Springer Berlin Heidelberg, 1995.

\bibitem{Kaxiras_2003}
E.~Kaxiras.
\newblock {\em Atomic and Electronic Structure of Solids}.
\newblock Cambridge University Press, Cambridge, 2003.

\bibitem{Koda2016}
D.~S. Koda, F.~Bechstedt, M.~Marques, and L.~K. Teles.
\newblock Coincidence lattices of {2D} crystals: Heterostructure predictions
  and applications.
\newblock {\em The Journal of Physical Chemistry C}, 120(20):10895--10908,
  2016.

\bibitem{Komsa2013}
H.-P. Komsa and A.~V. Krasheninnikov.
\newblock Electronic structures and optical properties of realistic transition
  metal dichalcogenide heterostructures from first principles.
\newblock {\em Phys. Rev. B}, 88:085318, Aug 2013.

\bibitem{Loh2015}
G.~C. Loh and R.~Pandey.
\newblock A graphene-boron nitride lateral heterostructure - a first-principles
  study of its growth{,} electronic properties{,} and chemical topology.
\newblock {\em J. Mater. Chem. C}, 3:5918--5932, 2015.

\bibitem{Mazzi:2011hl}
G.~Mazzi and B.~J. Leimkuhler.
\newblock {Dimensional Reductions for the Computation of Time--Dependent
  Quantum Expectations}.
\newblock {\em SIAM J. Sci. Comput.}, 33(4):2024--2038, Jan. 2011.

\bibitem{prodan2012}
E.~Prodan.
\newblock Quantum transport in disordered systems under magnetic fields: A
  study based on operator algebras.
\newblock {\em Appl. Math. Res. Express}, pages 176--255, 2013.

\bibitem{Roder:1997ee}
H.~R{\"o}der, R.~N. Silver, D.~A. Drabold, and J.~J. Dong.
\newblock {Kernel polynomial method for a nonorthogonal electronic-structure
  calculation of amorphous diamond}.
\newblock {\em Phys. Rev. B}, 55(23):15382--15385, June 1997.

\bibitem{Silver:1996uj}
R.~N. Silver, H.~Roeder, A.~F. Voter, and J.~D. Kress.
\newblock {Kernel polynomial approximations for densities of states and
  spectral functions}.
\newblock {\em J. Comp. Phys.}, 124(1):115--130, 1996.

\bibitem{Terrones2014}
H.~Terrones and M.~Terrones.
\newblock Bilayers of transition metal dichalcogenides: Different stackings and
  heterostructures.
\newblock {\em Journal of Materials Research}, 29:373--382, 2 2014.

\bibitem{2DPerturb15}
G.~A. Tritsaris, S.~N. Shirodkar, E.~Kaxiras, P.~Cazeaux, M.~Luskin,
  P.~Plech\'a\v{c}, and E.~Canc\`es.
\newblock Perturbation theory for weakly coupled two-dimensional layers.
\newblock {\em Journal of Materials Research}, to appear.

\bibitem{numeric_integration}
J.~A.~C. Weideman.
\newblock Numerical integration of periodic functions: A few examples.
\newblock {\em The American Mathematical Monthly}, 109(1):21--36, 2002.

\bibitem{kernel_poly}
A.~Wei\ss{}e, G.~Wellein, A.~Alvermann, and H.~Fehske.
\newblock The kernel polynomial method.
\newblock {\em Rev. Mod. Phys.}, 78:275--306, Mar 2006.

\bibitem{Zorzi2015}
A.~Zorzi.
\newblock An elementary proof for the equidistribution theorem.
\newblock {\em The Mathematical Intelligencer}, 37(3):1--2, 2015.

\end{thebibliography}
\end{document}